\documentclass{article}
\usepackage{graphicx} 
\usepackage{amsmath}
\usepackage{amsfonts}
\usepackage{amsthm}
\usepackage{amssymb}
\usepackage[margin=1in]{geometry}
\usepackage{xcolor}
\usepackage{appendix}
\usepackage{cite}
\usepackage{mathrsfs} 
 \usepackage{pgfplots} 
 \pgfplotsset{compat=1.17}
 \usepackage{subcaption}
 \usepackage{hyperref}
 \usepgfplotslibrary{groupplots}

\newcommand{\bbC}{\mathbb{C}}
\newcommand{\rmd}{\mathrm{d}}
\newcommand{\bbE}{\mathbb{E}}\newcommand{\rme}{\mathrm{e}}

\newcommand{\rmi}{\mathrm{i}}

\newcommand{\bbN}{\mathbb{N}}

\newcommand{\bbR}{\mathbb{R}}

\newcommand{\sfD}{\mathsf{D}}

\newcommand{\cD}{\mathcal{D}}

\newcommand{\cI}{\mathcal{I}}

\newcommand{\scrN}{\mathscr{N}}

\newcommand{\scrS}{\mathscr{S}}

\DeclareMathOperator{\supp}{supp}

\newtheorem{thm}{Theorem}

\newtheorem{prop}{Proposition}
\newtheorem{cor}{Corollary}
\newtheorem{lem}{Lemma}
\newtheorem{rem}{Remark}

\newtheorem{defi}{Definition}

\title{ On the Evolution of the Capacity-Achieving Input Support for the Amplitude-Constrained AWGN Channel}
\author{
Luca~Barletta%
\thanks{Luca Barletta is with the Dipartimento di Elettronica, Informazione e Bioingegneria, Politecnico di Milano, 20133 Milano, Italy (e-mail: luca.barletta@polimi.it).}
\and
Alex~Dytso%
\thanks{Alex Dytso is with Qualcomm Flarion Technology, Inc., Bridgewater, NJ 08807, USA (e-mail: odytso2@gmail.com).}
}

\begin{document}

\maketitle

\begin{abstract}
We consider an additive white Gaussian noise channel subject to a peak-amplitude constraint and study how the support of the capacity-achieving input distribution (CAID) evolves as the amplitude constraint~$A$ varies. Although the CAID is known to be unique, symmetric, discrete, and finitely supported, the structure of its support transitions has remained largely unresolved. We show that the origin is the only possible degenerate support point and the only possible inactive contact point of the KKT function. We then establish continuity properties of the optimal distribution and the KKT function and prove that every nonzero support point is locally stable: under small changes in $A$, it persists as a unique nearby atom whose location and probability mass vary continuously. Combining these results, we prove that, locally,  the support cardinality at a nearby amplitude is either unchanged or larger by one and that any such increase can occur only at the origin. 
 Consequently, all local changes in support cardinality are confined to the origin: locally, the only possible transition mechanisms are the appearance or disappearance of the point at the origin and the splitting or merging of the origin into a symmetric pair.
\end{abstract}

\section{Introduction}
We consider a discrete-time additive white Gaussian noise (AWGN) channel subject to a peak-amplitude constraint. The channel output is given by
\begin{equation}
Y = X + Z,
\end{equation}
where the input random variable $X$ satisfies $|X| \leq A$ almost surely (a.s.), and $Z$ is a standard normal random variable independent of $X$.

We are interested in the channel capacity, defined as
\begin{equation}
C(A) = \max_{X:\:|X|\leq A} I(X;Y),
\label{eq:Capacity_def}
\end{equation}
where $I(X;Y)$ denotes the mutual information between $X$ and $Y$. We denote by $X^\star$ a capacity-achieving input random variable and by $Y^\star$ the corresponding induced output random variable. In general, both the exact value of the capacity $C(A)$ and the precise structure of the capacity-achieving input distribution (CAID) $P_{X^\star}$ remain unknown.

It was established in \cite{smith1971information} that the CAID is discrete and has finite support. More recent results have established lower and upper bounds on the support cardinality of orders $A\sqrt{\log A}$ and $A^2$, respectively~\cite{wang2025improvedlowerboundcardinality}. Setting the cardinality question aside, in this work we are interested in how the support evolves as $A$ increases. Fig.~\ref{fig:pmf_evo} illustrates the evolution of the support as a function of $A$, as computed numerically.
\begin{figure}
\centering
\input{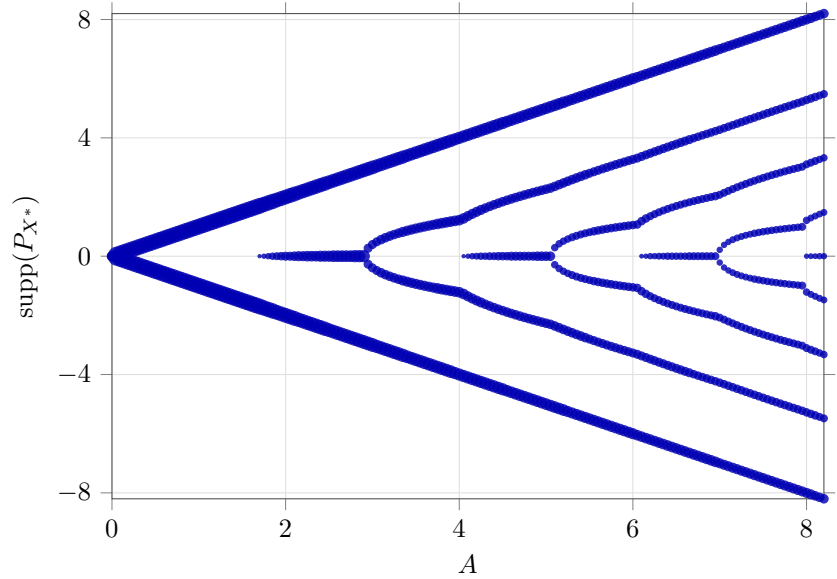}
\caption{Evolution of $\supp(P_{X^\star})$ as a function of $A$.  Each marker's size is proportional to its probability mass.}
\label{fig:pmf_evo}
\end{figure}
From Fig.~\ref{fig:pmf_evo}, we make the following observations:
\begin{enumerate}
\item The points $\pm A$ always belong to the support.
\item New support points always appear at zero.
\item The number of support points increases by at most one at each transition, either through the appearance of a new point or through the splitting of an existing point into two.
\item  Every nonzero support point is locally persistent and moves continuously under sufficiently small changes of $A$.
\end{enumerate}

The first claim was established by Smith in \cite{smith1971information}. 
 Claims~2 and~3 were investigated in
\cite{sharma2010transition} and conjectured to hold for all $A$.\footnote{Smith in~\cite[pp.~42--44]{smith1969Thesis} also implicitly makes this conjecture in the construction of his numerical procedure for finding the capacity-achieving distribution.}
In this work, we prove Claim~4 and establish local, unoriented forms of Claims~2 and~3. Specifically, we show that every nonzero support point has a unique local continuation, whereas all local changes in support cardinality are confined to the origin and can change the
cardinality by at most one. Determining the direction of these transitions as $A$ increases, and hence proving monotonicity of the support cardinality, remains open.

\subsection{Past Work} 
The literature on CAIDs is now extensive, and we do not attempt to provide a comprehensive review here. The interested reader is referred to~\cite{CISSdytso2018,dytso2019capacity,PoissonCardBounds} and the references therein.

The channel capacity problem in~\eqref{eq:Capacity_def} was first investigated by Shannon in his seminal work~\cite{Shannon:1948}, where he established the first upper and lower bounds on capacity. He further showed that, in the low-SNR regime, the capacity under a peak-amplitude constraint has the same asymptotic behavior as that under an average-power (second-moment) constraint.

Further progress in characterizing the capacity and the structure of the optimal input distribution was made by Smith~\cite{smith1969Thesis,smith1971information}. Smith proved that the CAID is unique, symmetric about the origin, and discrete with finite support. He also showed that, for every \(A<0.1\), the optimal distribution assigns equal probability to the two points \(\{\pm A\}\), thereby yielding an explicit capacity characterization in this regime. However, the precise scaling of the support size with the amplitude constraint remains unresolved. The best known non-asymptotic bounds, \cite{wang2025improvedlowerboundcardinality} and~\cite{dytso2017amplitude_longer}, place it between  $A\sqrt{\log A}$ and $A^2$, respectively.  

The work most closely related to ours is that of Sharma and Shamai~\cite{sharma2010transition}, who investigated transitions in the support of the CAID. They proved that the equiprobable binary distribution supported on $\{\pm A\}$ is optimal if and only if $A\leq \bar{A}\approx 1.665$, where $\bar{A}$ is characterized as the solution of an integral equation.  Furthermore, they demonstrated that a ternary input supported on $\{-A,0,A\}$ is capacity-achieving for all $\bar{A} \le A \le \bar{\bar{A}} \approx 2.9075$.\footnote{ The value reported in~\cite{sharma2010transition} is
$\bar{\bar{A}}\approx 2.786$. Our numerical computations suggest that this
value is inaccurate and that the transition from a three-point to a four-point
CAID instead occurs at
$\bar{\bar{A}}\approx 2.9075$. } Finally, Sharma and Shamai \cite{sharma2010transition} conjectured that
the support cardinality $K(A)$ is nondecreasing in $A$ and that it
increases by at most one at each transition. Their numerical
investigation further suggested that an additional mass point
originates at the center of the input interval. 
 In this work, we show that, relative to any fixed amplitude $A_0$, the support cardinality at all sufficiently nearby amplitudes is either $K(A_0)$ or $K(A_0)+1$. Moreover, every nonzero support point is
locally persistent, so all local birth/death and  splitting/merging events are confined to the origin. Our result does not determine the orientation of these transitions as $A$ increases, and the global monotonicity of $K(A)$ remains open.

Beyond questions concerning the form of the CAID, considerable attention has been devoted to bounding the capacity in~\eqref{eq:Capacity_def}. The principal techniques for deriving upper bounds can be organized into three broad classes. One approach invokes the maximum-entropy property~\cite[Chapter~12]{CoverInfoTheory}: appropriate moment constraints on the channel output are used to bound its differential entropy $h(Y)$, and hence the mutual information~\cite{shamai1990capacity,dytso2019amplitude}. A second approach follows from capacity duality. In this framework, an auxiliary output distribution is introduced, and the capacity is bounded through a relative-entropy minimization; evaluating the resulting expression at any admissible, though generally nonoptimal, output distribution produces an upper bound. This method underlies, among others, the McKellips bound~\cite{mckellips2004simple} and the bound of Thangaraj and Kramer~\cite{thangaraj2017capacity}; a detailed treatment of the duality method can be found in~\cite{lapidoth2003capacity}. A third line of analysis uses the integral relationship between mutual information and the minimum mean-square error (MMSE)~\cite{I-MMSE}. In this case, tractable upper bounds are obtained by evaluating the estimation error of a suitably chosen suboptimal estimator~\cite{dytso2017view}.

Related support-transition phenomena have also been studied in finite-alphabet rate-distortion problems. In particular, Agmon~\cite{agmon2023root} showed that optimal test channels typically follow piecewise-smooth solution paths interrupted by bifurcations, including cluster-vanishing and support-switching transitions. That work relies on implicit differentiation and root tracking of Blahut--Arimoto fixed points, whereas our approach studies the zeros and multiplicities of the KKT function for the amplitude-constrained Gaussian channel.

\subsection{Outline and Contributions}

The organization of the paper and its main contributions are as follows:
\begin{enumerate}
    \item Section~\ref{sec:preliminaries} develops the analytical
    preliminaries used throughout our proofs:
    \begin{itemize}
        \item Section~\ref{sec:KKT_conditions} recalls the KKT
        conditions and introduces the KKT function, which is our main
        tool for studying the structure of the optimal support. We also
        derive the higher-order derivatives of the KKT function and
        discuss its entire extension to the complex plane.

        \item Section~\ref{sec:classification_of_zeros} classifies the
        interior zeros of the KKT function into three categories.

        \item Section~\ref{sec:zero_counting_stability} reviews the
        zero-counting and stability tools used in our analysis,
        including Karlin's oscillation theorem, Rolle's theorem with
        multiplicities, and Hurwitz's theorem.
    \end{itemize}

    \item Section~\ref{sec:zeros_G_A} studies the structural properties
    of the zeros of the KKT function:
    \begin{itemize}
        \item Section~\ref{sec:bounds_on_number_zeros} revisits the
        bound on the number of zeros of the KKT function established
        in~\cite{dytso2019capacity}, while explicitly accounting for
        their multiplicities.

        \item Section~\ref{sec:mult_of_zeros} studies the multiplicities
        of the zeros of $G_A$. The main result of this subsection shows
        that the origin is the only possible degenerate interior support
        point and the only possible inactive contact point.
    \end{itemize}

    \item Section~\ref{sec:transition_points} studies how the support
    points evolve as $A$ varies:
    \begin{itemize}
        \item Section~\ref{sec:local_stability_of_non_degenerate}
        establishes several continuity results with respect to $A$,
        including the weak continuity of the CAID and the locally uniform continuity of the KKT
        function. The main result of this subsection shows that
        nondegenerate support points are locally stable: they move
        continuously and remain single atoms.

        \item Section~\ref{sec:transtion_of_sup_points} establishes one
        of our main results: locally, the number of support points can
        increase by at most one as $A$ varies.
    \end{itemize}
\end{enumerate}

 The code and numerical data used to reproduce all numerical figures in this paper are publicly available online \cite{BarlettaDytso2026Code}. The remainder of this section is used to define some relevant notation.  
\subsection{Notation}
Throughout the paper, random variables are denoted by uppercase letters and their realizations by the corresponding lowercase letters. The distribution of a random variable $X$ is denoted by $P_X$, and its support is defined as \begin{equation} \supp(P_X) = \left\{ x\in\bbR: P_X(\mathcal O)>0 \text{ for every open set }\mathcal O\ni x \right\}. \end{equation} The notation $|\cdot|$, depending on the context, denotes either absolute value or the cardinality of a set. Weak convergence of probability measures is denoted by $\Rightarrow$. The relative entropy between probability distributions $P$ and $Q$ is denoted by $\sfD(P\|Q)$, while $I(X;Y)$ and $h(X)$ denote mutual information and differential entropy, respectively.  All logarithms are natural. The standard Gaussian density is denoted by \begin{equation} \phi(x) = \frac{1}{\sqrt{2\pi}} \rme^{-x^2/2}, \qquad x\in\bbR. \end{equation} Finally, for $z\in\bbC$ and $r>0$, $B(z,r)$ and $\overline{B(z,r)}$ denote the open and closed disks of radius $r$ centered at $z$, respectively.

\section{Preliminaries}  
\label{sec:preliminaries}

In this section, we discuss some of the preliminary results and definitions needed in our analysis. 

\subsection{KKT Conditions}
\label{sec:KKT_conditions}

A starting tool for studying discrete distributions is the following set of KKT conditions \cite[Cor.~1]{smith1971information}, \cite[Thm.~10]{CISSdytso2018}.

\begin{lem}\label{lem:KKT}  The capacity-achieving distribution $P_{X^\star}$ and induced output distribution $P_{Y^\star}$ satisfy the following:  for $A>0$ let
\begin{equation}
    G_A(x)  =  \sfD( P_{Y|X}(\cdot| x) \| P_{Y^\star}) - C(A);
\end{equation}
Then,
\begin{align}
  G_A(x)  &\le 0 , \quad x \in [-A,A],  \label{eq:KKT_upper_bound} \\
   G_A(x) &= 0, \quad x\in {\rm supp}(P_{X^\star}). \label{eq:equality_eq_in_KKT}
\end{align}

\end{lem}

Lemma~\ref{lem:KKT} is the key to studying the support of $P_{X^\star}$. The main observation is that support points are related to zeros and maxima of the KKT function $G_A$. In particular, Lemma~\ref{lem:KKT} implies that
\begin{equation}
\supp(P_{X^\star})
\subseteq
\left \{x\in[-A,A]:\:G_A(x)=0\right \}.
\end{equation}
Consequently, studying the zeros of $G_A$ provides a natural way to investigate the support of $P_{X^\star}$. This connection was first observed by Smith in~\cite{smith1971information} and has since been exploited in \cite{dytso2019capacity} and related works to derive bounds on the cardinality of $\supp(P_{X^\star})$. Fig.~\ref{fig:GA_evolution} illustrates the behavior of $G_A$ vs.~$x$ for several values of $A$. 
Thus, to understand how $\supp(P_{X^\star})$ evolves with $A$, we must examine how the zeros of $G_A$ evolve with $A$. In the next subsection, we introduce the terminology and definitions concerning zeros that will be used throughout the paper.

\begin{rem}
    The KKT conditions in Lemma~\ref{lem:KKT} imply that interior support points are not merely zeros of $G_A$, but also local maxima. Most previous analyses exploit only the zero condition. A notable exception is \cite{Binomial_ISIT}, where the local maximality property was used to sharpen Witsenhausen's bound on the number of support points \cite{WitsenhausenBOund}. In this work, this property plays a central role in characterizing the  multiplicity of interior zeros.
\end{rem}

\begin{figure}
\begin{subfigure}[t]{0.46\textwidth}
    \centering
    \input{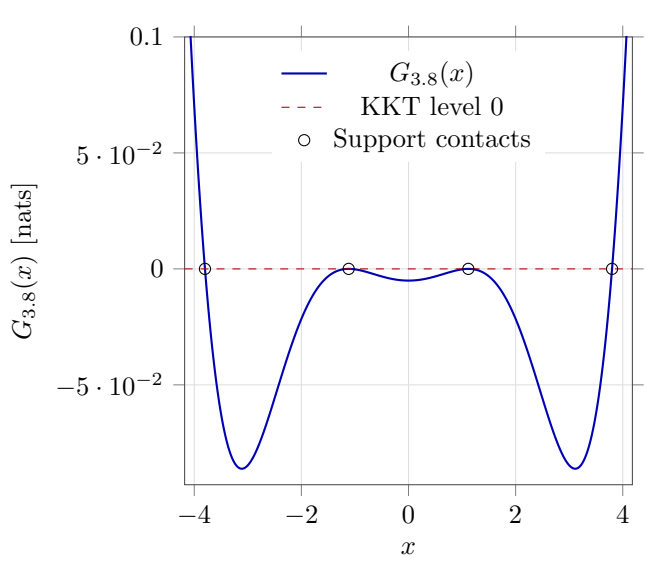}
    \caption{KKT function for $A=3.8$.}
    \label{fig:G_3d8}
\end{subfigure}
\hfill
\begin{subfigure}[t]{0.46\textwidth}
    \centering
    \input{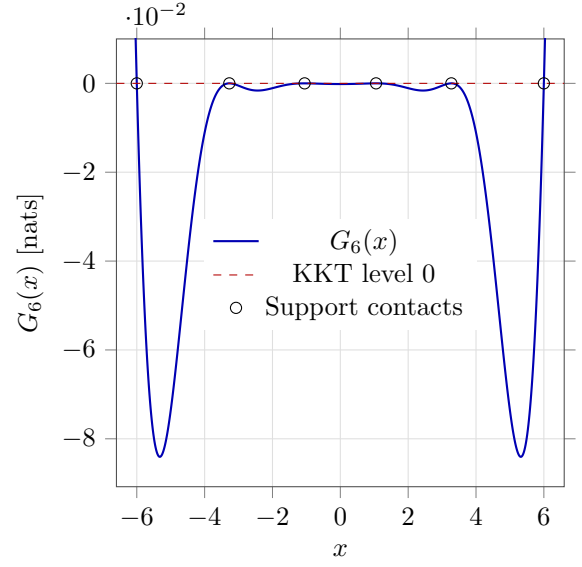}
    \caption{KKT function for $A=6$.}
    \label{fig:G_6}
\end{subfigure}

\par\medskip

\noindent\makebox[\linewidth][c]{%
    \begin{subfigure}[t]{0.46\linewidth}
        \centering
        \input{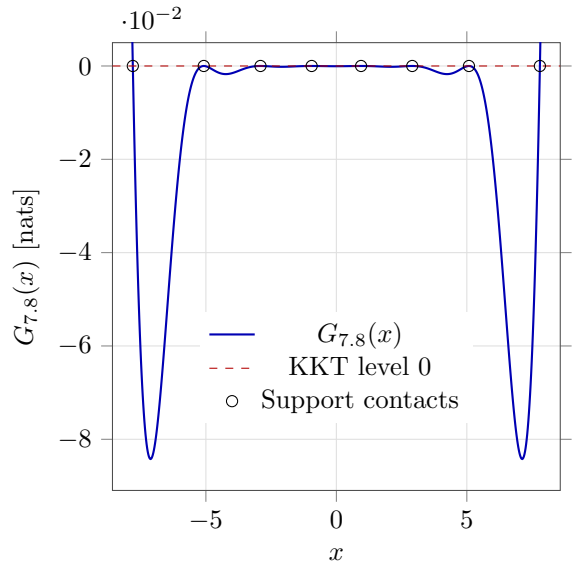}
        \caption{KKT function for $A=7.8$.}
        \label{fig:G_7d8}
    \end{subfigure}%
}
\caption{Evolution of the KKT function $G_A$ for different values of the amplitude constraint $A$.}
\label{fig:GA_evolution}
\end{figure}

Before concluding this section, we present some properties of $G_A$, which will be useful in our study of multiplicity. First, we characterize the structure of the derivatives of $G_A$.
\begin{lem}\label{lem:Derivative} For $A >0$ and $x \in \bbR,$
\begin{align}
    G'_A(x) &= x -\bbE \left[ \bbE[X^\star|Y^\star=x+Z] \right]  , \\
     G''_A(x) &= 1 -  \bbE \left[ {\rm Var}(X^\star|Y^\star=x+Z) \right], \\
     G^{(k)}_A(x) &=  -  \bbE \left[ \kappa_k(X^\star|Y^\star=x+Z) \right], \quad k \ge 3,
\end{align}
    where $\kappa_k(X^\star|Y^\star=y)$ is the conditional cumulant of order $k$. 
\end{lem}
\begin{proof} We begin by recalling the following result  \cite{dytso2022conditional}: for $y \in \bbR$
\begin{subequations}
    \begin{align}
 \frac{\rmd}{\rmd y }\log f_{Y}(y) &= \bbE[X|Y=y] -y, \\
  \frac{\rmd^2}{\rmd y^2 }\log f_{Y}(y) &= {\rm Var}(X|Y) -1 ,\\
  \frac{\rmd^k}{\rmd y^k }\log f_{Y}(y) &= \kappa_k(X|Y=y), \, k \ge 3.
\end{align}
\label{eq:derivatives}
\end{subequations}

Next, for $x \in \bbR$ note that
    \begin{align}
        \frac{\rmd^k}{\rmd x^k} G_A(x) &= \frac{\rmd^k}{\rmd x^k} \left( -\bbE \left[ \log f_{Y^\star}(x+Z)\right] - h(Z) - C(A)\right) \\
        &= - \bbE \left[  \frac{\rmd^k}{\rmd x^k}\log f_{Y^\star}(x+Z)\right] , \label{eq:der_G_ins}
        \end{align}
    where the interchange of expectation and derivative follows from Leibniz rule which can be invoked since for inputs  $X$ that  have bounded support, we have that \cite[Prop.~9]{dytso2022conditional}
    \begin{equation}
        \left|  \frac{\rmd^k}{\rmd y^k}\log f_{Y}(y)\right|  \le a_k |y|^k+ b_k, \label{eq:bound_derivative_log}
    \end{equation}
    for some positive constants $a_k,b_k$. The proof is concluded by applying \eqref{eq:derivatives} to \eqref{eq:der_G_ins}.
\end{proof}

Finally, we recall the following result of Smith \cite[App.~B]{smith1969Thesis} (see also \cite[Prop.~3]{tchamkerten2004discreteness}), which establishes that $G_A$ admits a well-defined extension to the complex plane.

\begin{lem}
\label{lem:entire_extension_GA}
Fix $A>0$. The function $G_A:\bbR\to\bbR$ admits a unique entire
extension to $\bbC$, which we also denote by $G_A$. This extension is
given by
\begin{equation}
    G_A(z)
    =
    -h(Z)-C(A)
    -
    \int_{\bbR}
        \phi(y-z)\log f_{Y}^\star(y)
    \rmd y,
    \qquad z\in\bbC,
    \label{eq:complex_extension_GA}
\end{equation}
where $\phi$ denotes the entire extension of the standard Gaussian
density:
\begin{equation}
    \phi(w)
    =
    \frac{1}{\sqrt{2\pi}}
    \exp\left(-\frac{w^2}{2}\right),
    \qquad w\in\bbC.
\end{equation}
The integral in \eqref{eq:complex_extension_GA} converges absolutely
for every $z\in\bbC$ and uniformly on every compact subset of $\bbC$.
\end{lem}

 We will also use the fact that $G_A$ is not identically zero. Indeed,
since $|X^\star|\leq A$ a.s.,
\begin{equation}
\mathbb E[X^\star\mid Y^\star]\leq A
\qquad\text{a.s.}
\end{equation}
Hence, by Lemma~\ref{lem:Derivative}, for every $x>A$,
\begin{equation}
G_A'(x)
=
x-\mathbb E\left[
\mathbb E[X^\star\mid Y^\star=x+Z]
\right]
\geq x-A>0.
\end{equation}
Thus, $G_A\not\equiv0$. Since $G_A$ is entire by
Lemma~\ref{lem:entire_extension_GA}, its zeros are isolated and,
consequently, every zero of $G_A$ has finite multiplicity.

\subsection{Classification of Zeros of $G_A$}
\label{sec:classification_of_zeros}
Our goal is to study properties of zeros of  $G_A$. Here we introduce a few important definitions that we will use.  We begin by recalling the definition of multiplicity of a zero. 
\begin{defi} Let $f$ be a real analytic function around $x_0$.  We say that the zero $x_0$ has \emph{multiplicity} $m$, or \emph{order} $m$, if
\begin{equation}
    f(x_0)= f'(x_0)= \ldots f^{(m-1)}(x_0) =0, \text{ and } f^{(m)}(x_0) \neq 0.
\end{equation}
We use the notation 
\begin{equation}
    {\rm ord}_{x_0} (f) = m. 
\end{equation}
if the zero $x_0$ has multiplicity $m$.  Additionally, 
\begin{equation}
    N_{\text{mul}}(\cI ,f)
\end{equation}
denotes the number of zeros of $f$, on an interval $\cI$, counted with multiplicity.  
    
\end{defi}

Another quantity that we will need is the number of sign changes. 
\begin{defi}  The number of sign changes of a function $\xi: \mathcal{X} \to \mathbb{R}$ is given by 
\begin{equation}
  \scrS(\xi) = \sup_{m\in \bbN } \left\{  \sup_{  \substack{y_1,\dots,y_m \in \mathcal{X} \\
  y_1< \cdots< y_m }  } \scrN \{ \xi (y_i) \}_{i=1}^m\right\} \text{,}
\end{equation}
where  $\scrN\{ \xi (y_i) \}_{i=1}^m$ is the number of sign changes  of the sequence $\{ \xi (y_i) \}_{i=1}^m $.
\end{defi} 

In this work, we classify the interior zeros of $G_A$ into three categories, as defined below.

\begin{defi}
The set of \emph{interior degenerate support points} is defined as
\begin{equation}
\cD_A
=
\left \{
x\in\supp(P_{X^\star})\cap(-A,A):\:
G_A''(x)=0
\right \}.
\end{equation}
The set of \emph{inactive contact points} is defined as
\begin{equation}
\cI_A
=
\left\{
x\in(-A,A)\setminus\supp(P_{X^\star}):\:
G_A(x)=0
\right\}. \label{eq:def_I_A}
\end{equation}
Thus, an inactive contact point is an interior point at which the KKT inequality holds with equality but to which the input distribution assigns no mass. Finally, an interior support point
\begin{equation}
x\in
\bigl(\supp(P_{X^\star})\cap(-A,A)\bigr)\setminus\cD_A
\end{equation}
is called a \emph{nondegenerate support point}. Equivalently, such a point satisfies $G_A''(x)\neq0$ and, in fact, $G_A''(x)<0$.
\end{defi}

\subsection{Zero-Counting Techniques and Stability of Zeros}
\label{sec:zero_counting_stability}

We now review some of the zero-bounding techniques that we will need.  

We begin by stating Karlin's oscillation theorem. The following theorem, shown in \cite[Thm.~3]{karlin1957polya} (see also  \cite[Thm.~3.1,~p.~21]{KarlinBook1968}), will be a key step in the proof of the upper bound on the number of mass points.

\begin{thm}
\label{thm:karlins_thm}
Let $f:\bbR\to\bbR$ have a finite number of sign changes, with
\begin{equation}
    \scrS(f)=n<\infty.
\end{equation}
Suppose that the Gaussian transform
\begin{equation}
    F(y)
    =
    \int_{\bbR}
        f(x)
        \exp\left(-\frac{(y-x)^2}{2}\right)
        \rmd x
\end{equation}
is well defined, belongs to $C^n(\bbR)$, and is not identically zero. Then
\begin{equation}
    N_{\mathrm{mul}}(\bbR,F)
    \leq
    \scrS(f).
\end{equation}
If $f$ is continuous and its zeros have well-defined finite
multiplicities, then
\begin{equation}
    \scrS(f)
    \leq
    N(\bbR,f)
    \leq
    N_{\mathrm{mul}}(\bbR,f),\label{eq:Karlins_bound}
\end{equation}
 where $N(\bbR,f)$ is the number of distinct zeros of $f$ in $\bbR$.
\end{thm}

Sometimes it will be convenient to count zeros of the derivative of the function, the classical tool here is Rolle's theorem, stated below in a form that accounts for multiplicities. 

\begin{lem}
\label{lem:roll_mulitplicity}
Let $f$ be a sufficiently smooth real-valued function on an interval
containing finitely many distinct zeros $x_1<x_2<\cdots<x_r$
on $\cI$ each with finite multiplicity.
Then, 
\begin{equation}
N_{\mathrm{mul}}( \cI, f')\geq N_{\mathrm{mul}}(  \cI, f)-1. \label{eq:Rolls_bound}
\end{equation}
\end{lem}

\begin{proof} The proof is given in \cite[Section~2.3, Exercise~21, p.~87]{TrenchRealAnalysis}. For completeness, we also provide the proof in Appendix~\ref{app:lem:roll_mulitplicity}. 

\end{proof}

We will also need the following results about zeros of exponential polynomials, which can be found in \cite[Prop.~28]{kounchev2021error}. 
\begin{lem}
Let
\begin{equation}
    f(t)=\sum_{k=1}^n P_k(t)e^{\lambda_k t},
\end{equation}
where $\lambda_1,\dots,\lambda_n\in\mathbb{R}$ are pairwise distinct and
each $P_k$ is a polynomial of degree $d_k$.

Then, $f$ is an exponential polynomial of order $n$, and the total number
of real zeros of $f$, counted with multiplicity, satisfies
\begin{equation}
    N_{\mathrm{mul}}(\bbR, f)
    \leq
    \sum_{k=1}^n (d_k+1) - 1. \label{Eq:expon_polynomials_bound}
\end{equation}

\end{lem}

Finally, we will need the following standard result known as Hurwitz’s theorem for
stability of zeros, including their multiplicities, under locally uniform convergence; see, e.g., \cite[Thm.~4]{Monard2015}.
\begin{lem} \label{lem:stability_lemma}
Let $z_0\in\mathbb C$ and $r>0$. Suppose that $f_n$ and $f$ are
holomorphic on an open set containing the closed disk
\begin{equation}
\overline{B(z_0,r)}
=
\left\{
z\in\mathbb C:
|z-z_0|\leq r
\right\},
\end{equation}
and that
\begin{equation}
f_n \to  f
\end{equation}
uniformly on $\overline{B(z_0,r)}$.

Assume that $z_0$ is the only zero of $f$ in
$\overline{B(z_0,r)}$ and that it has multiplicity $m$. Then, for all
sufficiently large $n$, the function $f_n$ has exactly $m$ zeros in
$B(z_0,r)$, counted according to multiplicity.
\end{lem}

\section{Properties of Zeros of $G_A$}
\label{sec:zeros_G_A}
In this section, we establish several structural properties of the zeros of $G_A$. In particular, the main result of this section is the complete characterization of the multiplicity of all zeros of $G_A$.

\subsection{Bounds on the Number of Zeros of $G_A$} 
\label{sec:bounds_on_number_zeros}

In this subsection, we  re-derive  the bound on the number of zeros of $G_A$   shown in \cite[Thm.~1]{dytso2019capacity}, while explicitly accounting for their multiplicities.

\begin{lem}
\label{lem:upper_bound_zeros_GA}
For every $A>0$, let $K(A)
=
\left|\supp(P_{X^\star})\right|$. Then,
\begin{equation}
N_{\mathrm{mul}}(\bbR,G_A)
\leq
2K(A).
\label{eq:upper_bound_zeros_GA_lemma}
\end{equation}
\end{lem}

\begin{proof}
Let
\begin{equation}
\supp(P_{X^\star})
=
\{x_1,\ldots,x_{K(A)}\},
\qquad
p_j=P_{X^\star}({x_j}).
\end{equation}
The induced output density can be written as
\begin{equation}
f_{Y^\star}(y)
=
\frac{\rme^{-y^2/2}}{\sqrt{2\pi}}
\sum_{j=1}^{K(A)}a_j\rme^{x_jy},
\qquad
a_j=p_j\rme^{-x_j^2/2}.
\label{eq:output_density_exponential_sum}
\end{equation}
To express the KKT function as a Gaussian transform, define
\begin{equation}
\kappa_A
=
\exp\bigl(-C(A)-h(Z)\bigr)
=
\frac{\rme^{-C(A)}}{\sqrt{2\pi e}},
\label{eq:def_kappa_A}
\end{equation}
where $h(Z)=\frac{1}{2}\log(2\pi e)$. By the definition of $G_A$,
\begin{align}
G_A(x)
&=
\sfD( P_{Y|X}(\cdot| x) \| P_{Y^\star})-C(A)\\
&=
-h(Z)-C(A)
-
\int_{\bbR}
\phi(y-x)\log f_{Y^\star}(y)\rmd y\\
&=
\int_{\bbR}
\phi(y-x)
\log\left(\frac{\kappa_A}{f_{Y^\star}(y)}\right)
\rmd y,
\label{eq:GA_gaussian_transform}
\end{align}
where
\begin{equation}
\phi(t)
=
\frac{1}{\sqrt{2\pi}}\rme^{-t^2/2}.
\end{equation}
Thus, $G_A$ is the Gaussian transform of
\begin{equation}
\xi_A(y)
=
\log\left(\frac{\kappa_A}{f_{Y^\star}(y)}\right),
\end{equation}
 and recall that $G_A \not\equiv 0$.
Since $ \xi_A  $ is continuous and   grows at most polynomially  as
$|y|\to\infty$, see \eqref{eq:bound_derivative_log},   its Gaussian transform is well defined.  Moreover, 
 since $\kappa_A\in(0,1/\sqrt{2\pi e}]$, 
\cite[Lemma~2]{dytso2019capacity} implies that
$f_{Y^\star}-\kappa_A$ has finitely many real zeros.
Consequently, $\xi_A$ has finitely many sign changes, and Karlin's
oscillation theorem (i.e., Theorem~\ref{thm:karlins_thm} can be invoked:
\begin{align}
N_{\mathrm{mul}}(\bbR,G_A)
&\leq
\scrS(\xi_A)\\
&\leq
N_{\mathrm{mul}}(\bbR,\xi_A)\\
&=
N_{\mathrm{mul}}\bigl(\bbR,f_{Y^\star}-\kappa_A\bigr).
\label{eq:using_karlin}
\end{align}
The last equality holds because $\xi_A$ and $f_{Y^\star}-\kappa_A$ have the same zeros with the same multiplicities.

Differentiating \eqref{eq:output_density_exponential_sum}, we obtain
\begin{equation}
f_{Y^\star}'(y)
=
\frac{\rme^{-y^2/2}}{\sqrt{2\pi}}Q_A(y),
\end{equation}
where
\begin{equation}
Q_A(y)
=
\sum_{j=1}^{K(A)}
a_j(x_j-y)\rme^{x_jy}.
\end{equation}
Moreover, since $f_{Y^\star}-\kappa_A$ is real analytic and has
finitely many real zeros, each of its zeros has finite multiplicity.
Therefore, Rolle's theorem with multiplicities gives
\begin{align}
N_{\mathrm{mul}}\bigl(\bbR,f_{Y^\star}-\kappa_A\bigr)
&\leq
N_{\mathrm{mul}}\bigl(\bbR,f_{Y^\star}'\bigr)+1\\
&=
N_{\mathrm{mul}}(\bbR,Q_A)+1.
\label{eq:using_rolls_multiplicyt}
\end{align}
The equality holds because the factor
$\rme^{-y^2/2}/\sqrt{2\pi}$ is strictly positive and does not affect the zeros or their multiplicities.

Finally, $Q_A$ is an exponential polynomial with $K(A)$ distinct exponents, and each polynomial coefficient $a_j(x_j-y)$ has degree one. Therefore, \eqref{Eq:expon_polynomials_bound} gives
\begin{equation}
N_{\mathrm{mul}}(\bbR,Q_A)
\leq
2K(A)-1.
\end{equation}
Combining the preceding bounds, we obtain
\begin{equation}
N_{\mathrm{mul}}(\bbR,G_A)
\leq
2K(A),
\end{equation}
which concludes the proof.
\end{proof}

\subsection{Order of Zeros of $G_A$}
\label{sec:mult_of_zeros}

In this subsection, we establish one of our main results: the origin is the only possible interior degenerate support point. We begin by characterizing the multiplicities of the different types of zeros of $G_A$.

\begin{prop} \label{prop:multiplicity} Fix some $A>0$. Then, the following properties hold.
\begin{enumerate}
   
    \item Suppose that $x_0$ is an interior zero of $G_A$. Then, there exists a positive integer $m$, such that
    \begin{equation}
        {\rm ord}_{x_0} (G_A) =2 m,
    \end{equation}
    i.e., interior zeros have even multiplicity. 
    \item $ \pm A$  are simple zeros of $G_A$, that is
    \begin{equation}
        {\rm ord}_{A} (G_A) = {\rm ord}_{-A} (G_A) =1; \label{eq:multiplicity_pm_A}
    \end{equation}
    \item if $x_0 \in \cD_A$, then
    \begin{equation}
        {\rm ord}_{x_0} (G_A)  \ge 4. \label{eq:multiplicity_of_D_A}
    \end{equation}
\item    if $x_0 \in  ( \supp(P_{X^\star}) \cap (-A,A) ) \setminus \cD_A$ (i.e., $x_0$ is an interior nondegenerate support point), then
    \begin{equation}
        {\rm ord}_{x_0} (G_A) = 2; \label{eq:mult_non_degenerate_points}
    \end{equation}
    \item  if $x_0 \in \cI_A$, then
    \begin{equation}
        {\rm ord}_{x_0} (G_A) \ge 2.  \label{eq:kissing_points}
    \end{equation}
 
\end{enumerate}

\end{prop}
\begin{proof}
    To see the first statement suppose that there exists an interior zero $x_0$ of $G_A$ that has an odd multiplicity. Then, in the vicinity of $x_0$, $G_A(x)$ would have a sign change. In other words, $x_0$ would be a zero of $G_A$  but not a maximum of $G_A$. This would contradict the KKT condition in \eqref{eq:KKT_upper_bound} (i.e., $G_A(x_0) \le 0$). 

    To see the second statement, note that by using Lemma~\ref{lem:Derivative}, and  $\bbE[X^\star|Y^\star] < A$, we have that 
    \begin{align}
        G'_A(A)  &= A - \bbE \left[ \bbE[X^\star|Y^\star=A +Z]  \right] \\
        &>0.
    \end{align}
 Similarly, using $\bbE[X^\star|Y^\star] > -A$, 
 \begin{equation}
        G'_A(-A) <0.  
    \end{equation}

 To show the third property, note that if $x_0 \in \cD_A$, then
    \begin{equation}
        G_A(x_0) = G_A'(x_0) = G_A''(x_0)=0 .
    \end{equation}
    Thus, $x_0 \in \cD_A$ has multiplicity at least three. However, by the first property, the multiplicity can only be even and, therefore, $x_0$ must have multiplicity at least four. 

    The fourth property follows since by the KKT conditions an interior nondegenerate support  point is also a maximum and by assumption $G_A''(x_0) \neq 0$. 

    The fifth property follows from the definition of $\cI_A$ in \eqref{eq:def_I_A} and the fact that multiplicity can only be even. This concludes the proof. 
\end{proof}

\begin{rem}

Claims 1, 3, 4, and 5 of Proposition~\ref{prop:multiplicity}
follow only from the KKT inequality and the analyticity of $G_A$.
Consequently, they remain valid for any channel whose KKT function is
real analytic and not identically zero. Claim~2 is also likely provable for a fairly large family of additive channels (e.g., noise distribution with log-concave probability density).  
\end{rem}

We are now ready to show the main result of this section. 
\begin{thm} Fix some $A>0$.  Then,
\begin{equation}
    \cD_A \cup \cI_A \subseteq \{ 0\}.
\end{equation}
Consequently, for every $A>0$, $|\cD_A| \le 1$.
\label{thm:size_of_degenerate_solution}
\end{thm}
\begin{proof}
Let $K(A) = | \supp(P_{X^\star})|$.  Then,  from Proposition~\ref{prop:multiplicity}, we have that  points $\pm A$ each have multiplicity one as shown in \eqref{eq:multiplicity_pm_A}, each nondegenerate interior point contributes multiplicity $2$ as shown in \eqref{eq:mult_non_degenerate_points},  points in $\cI_A$ contribute multiplicity at least two as shown in \eqref{eq:kissing_points}, and points in $\cD_A$ contribute multiplicity at least four as shown in \eqref{eq:multiplicity_of_D_A}. Therefore, 
    \begin{align}
N_{\text{mul}} (\bbR, G_A)  &\ge 2 +  2 ( K(A) - 2 -| \cD_A| )  + 4| \cD_A| + 2 |\cI_A|\\
& = 2 K(A) -2 +2 | \cD_A| + 2 |\cI_A|. \label{eq:lower_bound_num_zeros_of_n_g_a}
    \end{align}

Combining \eqref{eq:lower_bound_num_zeros_of_n_g_a} and \eqref{eq:upper_bound_zeros_GA_lemma}, we arrive at
\begin{align}
2 K(A) -2 +2 | \cD_A| + 2 |\cI_A|
\le  N_{\text{mul}}(\bbR, G_A)  \le 2 K(A)
\end{align}
which implies that
\begin{equation}
    | \cD_A| +  |\cI_A| \le 1.  \label{eq:cardinality_bound_on_D+I}
\end{equation}

Finally, since $P_{X^\star}$ is unique and symmetric, we have that  $\supp(P_{X^\star})$ is symmetric about the origin.
Moreover, $G_A$ and $G_A''$ are even functions. Consequently, both
$\cD_A$ and $\cI_A$ are symmetric, and hence
\begin{equation}
    x\in\cD_A\cup\cI_A
    \quad\Longrightarrow\quad
    -x\in\cD_A\cup\cI_A.
\end{equation}
Thus, if $\cD_A\cup\cI_A$ contained a nonzero point, it would contain
at least two distinct points, contradicting
\eqref{eq:cardinality_bound_on_D+I}. Hence,
\begin{equation}
    \cD_A\cup\cI_A\subseteq\{0\}.
\end{equation}
This concludes the proof.
\end{proof}

In Theorem~\ref{thm:size_of_degenerate_solution} we have shown that zero is the only possible degenerate support point. We now examine the multiplicity at zero more closely.  
\begin{cor}
\label{cor:order_at_origin}
Fix $A>0$. Then, the following properties hold:
\begin{itemize}
    \item If $0\notin\supp(P_{X^\star})$
        and $
        G_A(0)=0$,  then
    \begin{equation}
        {\rm ord}_0 (G_A)=2.
    \end{equation}

    \item If $   0\in\cD_A,$  then
    \begin{equation}
        {\rm ord}_0 (G_A)=4
    \end{equation}
    and
    \begin{equation}
        G_A^{(4)}(0)<0.
    \end{equation}
\end{itemize}
\end{cor}

\begin{proof}
Let $ K(A)
    =
    \left|
        \supp(P_{X^\star})
    \right|$. The proof of Theorem~\ref{thm:size_of_degenerate_solution} shows that
\begin{equation}
    N_{\mathrm{mul}}(\bbR,G_A)
    \leq
    2K(A).
    \label{eq:cor_zero_upper_bound}
\end{equation}
On the other hand, the two boundary support points contribute
multiplicity two in total, while the $K(A)-2$ interior support points
contribute multiplicity at least two each. Thus, before accounting for
any inactive contact point or additional degeneracy, the support points
already contribute at least
\begin{equation}
    2+2\bigl(K(A)-2\bigr)
    =
    2K(A)-2,
    \label{eq:cor_baseline_multiplicity}
\end{equation}
to $N_{\mathrm{mul}}(\bbR,G_A)$. Hence, only two additional units of
multiplicity are available. 

Suppose first that
\begin{equation}
    0\notin\supp(P_{X^\star})
    \qquad\text{and}\qquad
    G_A(0)=0.
\end{equation}
Since $x=0$ is an interior zero of $G_A$, Proposition~\ref{prop:multiplicity}
implies that its multiplicity is even and at least two. Since $0\notin\supp(P_{X^\star})$, this zero is
not included in the support contribution
\eqref{eq:cor_baseline_multiplicity}. Since only two additional units
of multiplicity are available, we obtain
\begin{equation}
    {\rm ord}_0 (G_A)=2.
\end{equation}

Now, suppose that 
\begin{equation}
    0\in\cD_A.
\end{equation}
In the baseline count \eqref{eq:cor_baseline_multiplicity}, the origin has
already been counted as an interior support point of multiplicity two.
By Proposition~\ref{prop:multiplicity}, degeneracy implies
\begin{equation}
    {\rm ord}_0 (G_A)\geq4.
\end{equation}
Thus, degeneracy requires at least two additional units of multiplicity
beyond the baseline contribution. On the other hand, \eqref{eq:cor_zero_upper_bound} and \eqref{eq:cor_baseline_multiplicity} show that at most two additional units of multiplicity are available in total. Therefore, the origin must contribute exactly two additional units beyond its baseline multiplicity two, and hence
\begin{equation}
    {\rm ord}_0 (G_A)=4.
\end{equation}

Finally, since, in view of \eqref{eq:KKT_upper_bound}, $G_A(x)\leq0$ on $[-A,A]$ and $G_A(0)=0$, the origin is
a local maximum of $G_A$. The order-four Taylor expansion gives
\begin{equation}
    G_A(x)
    =
    \frac{G_A^{(4)}(0)}{4!}x^4
    +
    o(x^4),
    \qquad
    x\to0.
\end{equation}
Because $G_A(x)\leq0$ in a neighborhood of zero and
$G_A^{(4)}(0)\neq0$, we must have
\begin{equation}
    G_A^{(4)}(0)<0.
\end{equation}
This concludes the proof. 
\end{proof}

Corollary~\ref{cor:order_at_origin} already strongly constrains the
possible local configurations at the origin. 
If the origin is an inactive contact point, then it is a zero of multiplicity two; if the origin is a degenerate support point, then it is a zero of multiplicity four. These multiplicities suggest that, under a sufficiently small perturbation of $A$, at most one or two support points, respectively, can occur near the origin. Together with symmetry of the CAID, the
corresponding local configurations would consist of either the origin itself or a symmetric pair ${\pm x_A}$. To make these conclusions rigorous, however, we must first establish continuity of the KKT function with respect to $A$ and then invoke stability of zeros under
holomorphic perturbations. This is done in the next  section.

Before concluding this section we demonstrate some possible transitions and values of the derivatives. Figure~\ref{fig:appearance_transition} numerically illustrates the
appearance of a new support point at the origin. For $A=A_-$, the KKT
inequality is strict in a neighborhood of the origin, and in particular
$G_{A_-}(0)<0$. At the critical amplitude $A=A_0$, the KKT function first
touches zero at the origin, while the capacity-achieving input assigns no
mass to this point. Thus, the origin is an inactive contact point. Since
\begin{equation}
    G_{A_0}''(0)<0,
\end{equation}
this contact has multiplicity two. For $A=A_+$, the origin becomes an
active support point, indicated by the filled circle. Notice that the
second derivative remains strictly negative throughout this transition,
so the newly appearing support point is nondegenerate.

Figure~\ref{fig:splitting_transition} numerically illustrates the second
possible transition mechanism. For $A=A_-$, the origin is a nondegenerate
support point, with
\begin{equation}
    G_{A_-}''(0)<0.
\end{equation}
At the critical amplitude $A=A_0$, the curvature at the origin vanishes:
\begin{equation}
    G_{A_0}''(0)=0.
\end{equation}
At the same time, the fourth derivative is strictly negative,
\begin{equation}
    G_{A_0}^{(4)}(0)<0,
\end{equation}
which is consistent with the origin being a zero of multiplicity four.
For $A=A_+$, the contact at the origin disappears and is replaced by two
symmetric nondegenerate support points. Thus, the support cardinality
increases by exactly one across this transition.

\begin{figure}
\begin{subfigure}[t]{\textwidth}
    \centering
\input{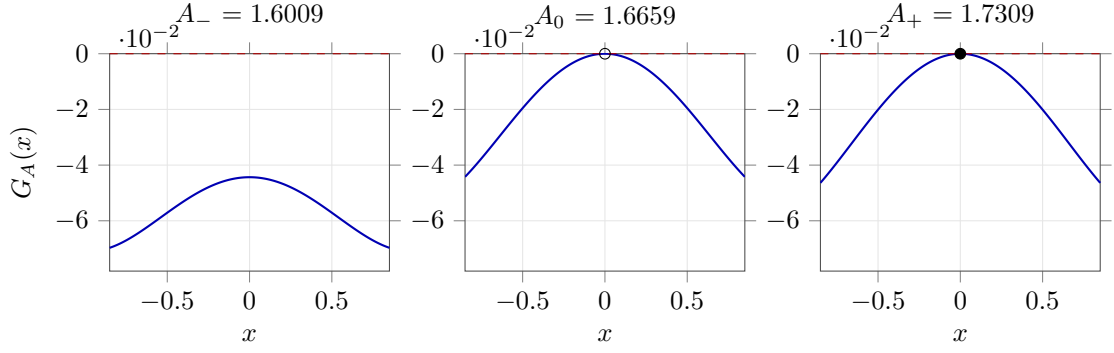}
\caption{KKT function.}
\end{subfigure}
\begin{subfigure}[t]{\textwidth}
    \centering
\input{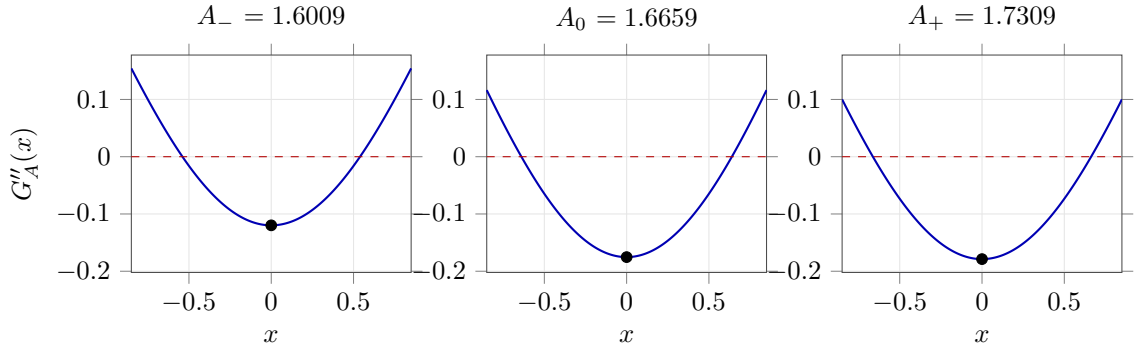}
\caption{Second derivative of the KKT function.}
\end{subfigure}
\begin{subfigure}[t]{\textwidth}
    \centering
\input{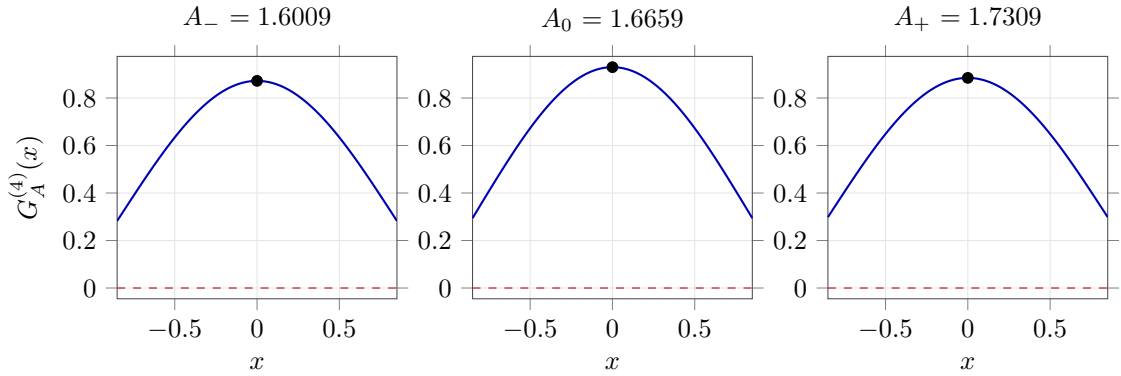}
\caption{Fourth derivative of the KKT function.}
\end{subfigure}
\caption{Numerical illustration of the appearance of a new support point
at the origin near the critical amplitude $A_0\approx1.6659$.}
\label{fig:appearance_transition}
\end{figure}

\begin{figure}
\begin{subfigure}[t]{\textwidth}
    \centering
\input{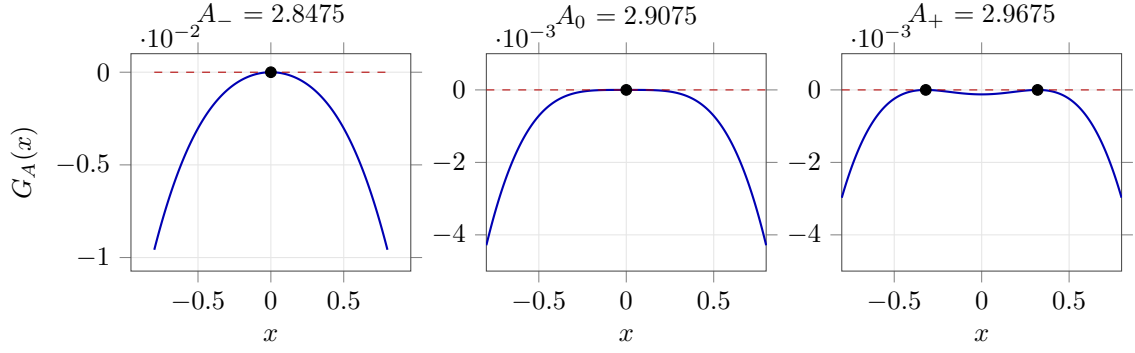}
\caption{KKT function.}
\end{subfigure}
\begin{subfigure}[t]{\textwidth}
    \centering
\input{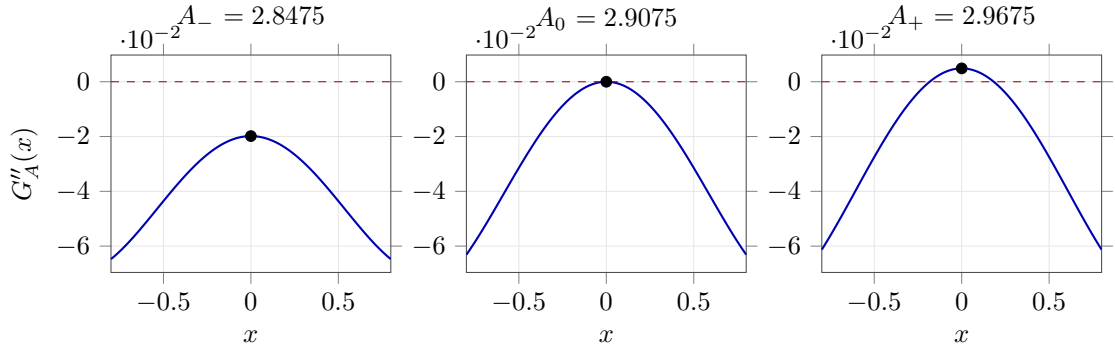}
\caption{Second derivative of the KKT function.}
\end{subfigure}
\begin{subfigure}[t]{\textwidth}
    \centering
\input{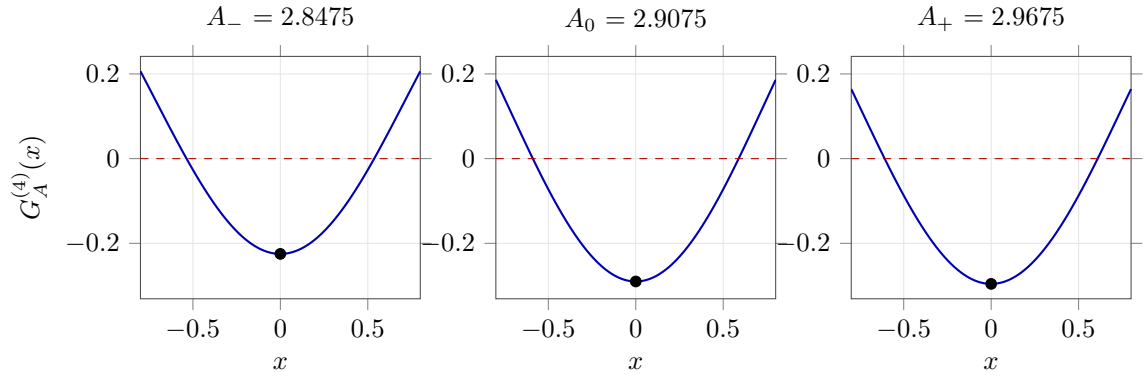}
\caption{Fourth derivative of the KKT function.}
\end{subfigure}
\caption{Numerical illustration of the splitting of the support point at
the origin into two symmetric support points near the critical amplitude
$A_0\approx2.9075$.}
\label{fig:splitting_transition}

\end{figure}

\section{Transitions of the Support Points}
\label{sec:transition_points}

In this section, we study how the support points behave as we change $A$. We make the dependence on the amplitude constraint explicit.
For each $A>0$, let $P_A^\star$ denote the CAID under the constraint $|X|\leq A$, and let
$X_A^\star\sim P_A^\star$. We write
\begin{equation}
    Y_A^\star = X_A^\star + Z
\end{equation}
and denote the corresponding output distribution and density by
$P_{Y,A}^\star$ and $f_{Y,A}^\star$, respectively. When the amplitude
$A$ is fixed and clear from the context, we suppress this dependence and
write $X^\star$, $Y^\star$, $P_{X^\star}$, and $f_{Y^\star}$.

\subsection{Local Stability of Nondegenerate Support Points}
\label{sec:local_stability_of_non_degenerate}

  In this subsection, we show that nondegenerate support points are locally stable under changes in the amplitude constraint. In particular, each such point persists as a unique nearby support point whose location and probability mass vary continuously with $A$. We start with the following definitions. 

  \begin{defi}[Weak Convergence] A sequence of probability measures $\{P_n\}$ on $\mathbb R$ converges weakly to a probability measure $P$, denoted by \begin{equation} P_n \Rightarrow P, \end{equation} if \begin{equation} \int_{\mathbb R} g(x)\,P_n(\rmd x) \to  \int_{\mathbb R} g(x)\,P(\rmd x) \end{equation} 
  for every bounded continuous function $g:\mathbb R\to\mathbb R$.
  \end{defi}

    \begin{defi}[Locally Uniform Convergence]  A sequence of functions $\{f_n\}$ converges locally uniformly to a function $f$ on $\mathbb R$ if, for every compact set $\mathcal K\subset\mathbb R$, \begin{equation} \sup_{y\in\mathcal K}|f_n(y)-f(y)| \to 0.  \end{equation} \end{defi}

We now establish the following continuity results for the optimal input and output distributions. 
\begin{prop} \label{prop:continuity_vs_A}
    If $A_n  \to  A_0>0,$ then
    \begin{enumerate}
\item \emph{Continuity of the Input and Output Distributions and Capacity:}
\begin{equation}
P_{A_n}^\star
\Rightarrow
P_{A_0}^\star,
\qquad
P_{Y,A_n}^\star
\Rightarrow
P_{Y,A_0}^\star,
\end{equation}
and
\begin{equation}
C(A_n)\to C(A_0).
\end{equation}
        
        \item   \emph{Convergence of the Output Densities:} 
        \begin{equation}
    f_{Y,A_n}^\star
    \to 
    f_{Y,A_0}^\star,
\end{equation}
where the convergence is locally uniform.
\item  \emph{Convergence of $G_A''$:}
\begin{equation}
    G_{A_n}''
    \to 
    G_{A_0}'',
\end{equation}
where the convergence is locally uniform.
\item \emph{Holomorphic continuity of $G_A$:}  Viewing each $G_A$ as the entire function $G_A:\bbC\to\bbC$ given by Lemma~\ref{lem:entire_extension_GA}, we have, for every compact set $\mathcal K\subset\bbC$, \begin{equation} \sup_{z\in\mathcal K} \left|G_{A_n}(z)-G_{A_0}(z)\right| \to 0. \end{equation}
    \end{enumerate}

\end{prop}
\begin{proof}
See Appendix~\ref{app:proof_continuity}.
\end{proof}

We now show that nondegenerate points are locally stable in the sense that they move continuously and remain a single atom. 

\begin{prop} \label{prop:contininous_motion}
Fix $A_0>0$, and let $x_0
    \in
    \supp(P_{A_0}^\star)\cap(-A_0,A_0)$ be a nondegenerate support point (i.e., $ G_{A_0}''(x_0)<0$).  Then, there exist $\delta>0$ and $\varepsilon>0$ such that, for every
$A$ satisfying
\begin{equation}
    |A-A_0|<\varepsilon,
\end{equation}
the interval $(x_0-\delta,x_0+\delta)$ contains exactly one point $x_A$ of $\supp(P_A^\star)$. Moreover, as $A\to A_0$, 
\begin{equation}
    x_A \to x_0
\end{equation}
and
\begin{equation}
    P_A^\star(\{x_A\})
     \to 
    P_{A_0}^\star(\{x_0\}). 
\end{equation}

\end{prop}

\begin{proof}
By claims~1 and~3 of Proposition~\ref{prop:continuity_vs_A}, we have
\begin{equation}
    P_A^\star\Rightarrow P_{A_0}^\star
    \qquad\text{as}\qquad A\to A_0
    \label{eq:weak_continuity_used_persistence}
\end{equation}
and
\begin{equation}
    G_A''\to G_{A_0}''
    \qquad\text{locally uniformly as}\qquad A\to A_0.
    \label{eq:second_derivative_continuity}
\end{equation}

Since $x_0$ is an isolated interior support point and
$G_{A_0}''(x_0)<0$, we may choose $\delta>0$ and $\eta>0$ such that
\begin{equation}
    [x_0-\delta,x_0+\delta]\subset(-A_0,A_0),
    \qquad
    \supp(P_{A_0}^\star)\cap[x_0-\delta,x_0+\delta]
    =
    \{x_0\},
    \label{eq:isolated_support_point_interval}
\end{equation}
and
\begin{equation}
    G_{A_0}''(x)\leq-2\eta,
    \qquad
    x\in[x_0-\delta,x_0+\delta].
    \label{eq:strict_concavity_at_A0}
\end{equation}
By \eqref{eq:second_derivative_continuity}, for all $A$ sufficiently
close to $A_0$,
\begin{equation}
    \sup_{x\in[x_0-\delta,x_0+\delta]}
    \left|G_A''(x)-G_{A_0}''(x)\right|
    <\eta.
    \label{eq:second_derivative_uniform_bound}
\end{equation}
Therefore,
\begin{equation}
    G_A''(x)\leq-\eta,
    \qquad
    x\in[x_0-\delta,x_0+\delta],
    \label{eq:strict_concavity_local_persistence}
\end{equation}
and hence $G_A$ is strictly concave on this interval. By taking $A$
sufficiently close to $A_0$, we may also ensure that
\begin{equation}
    [x_0-\delta,x_0+\delta]\subset(-A,A).
\end{equation}

It follows from \eqref{eq:isolated_support_point_interval} that the
boundary of $(x_0-\delta,x_0+\delta)$ has zero
$P_{A_0}^\star$-mass. Thus, by
\eqref{eq:weak_continuity_used_persistence},
\begin{equation}
    P_A^\star\bigl((x_0-\delta,x_0+\delta)\bigr)
    \to
    P_{A_0}^\star(\{x_0\})
    >0.
    \label{eq:local_mass_convergence}
\end{equation}
Consequently, $(x_0-\delta,x_0+\delta)$ contains at least one support
point of $P_A^\star$ for all $A$ sufficiently close to $A_0$.

We next show that it contains at most one such point. Suppose that
$a<b$ are two support points in this interval. By the KKT equality,
\begin{equation}
    G_A(a)=G_A(b)=0.
\end{equation}
However, strict concavity gives
\begin{equation}
    G_A\bigl(\lambda a+(1-\lambda)b\bigr)
    >
    \lambda G_A(a)+(1-\lambda)G_A(b)
    =
    0,
    \qquad \lambda\in(0,1),
\end{equation}
which contradicts the KKT inequality in \eqref{eq:KKT_upper_bound} (i.e., $G_A\leq0$ on $[-A,A]$).
Therefore, $(x_0-\delta,x_0+\delta)$ contains exactly one support point
of $P_A^\star$, which we denote by $x_A$.

To show that $x_A\to x_0$, fix any $r\in(0,\delta)$. The same weak
convergence argument gives
\begin{equation}
    P_A^\star\bigl((x_0-r,x_0+r)\bigr)
    \to
    P_{A_0}^\star(\{x_0\})
    >0.
\end{equation}
Hence, for all $A$ sufficiently close to $A_0$, the interval
$(x_0-r,x_0+r)$ contains a support point. Since $x_A$ is the unique
support point in $(x_0-\delta,x_0+\delta)$, it follows that
\begin{equation}
    |x_A-x_0|<r.
\end{equation}
Since $r>0$ is arbitrary,
\begin{equation}
    x_A \to x_0.
\end{equation}

Finally, since $x_A$ is the only support point in
$(x_0-\delta,x_0+\delta)$, \eqref{eq:local_mass_convergence} gives
\begin{equation}
    P_A^\star(\{x_A\})
    =
    P_A^\star\bigl((x_0-\delta,x_0+\delta)\bigr)
    \to
    P_{A_0}^\star(\{x_0\}).
\end{equation}
This concludes the proof.
\end{proof}

\begin{prop}
\label{prop:boundary_continuity}
Fix $A_0>0$. Then, as $A\to A_0$,
\begin{equation}
P_A^\star({A})
\to
P_{A_0}^\star({A_0}),
\qquad
P_A^\star({-A})
\to
P_{A_0}^\star({-A_0}).
\end{equation}
Moreover, for $A$ sufficiently close to $A_0$, the points $A$ and
$-A$ are the unique support points of $P_A^\star$ in sufficiently
small neighborhoods of $A_0$ and $-A_0$, respectively.
\end{prop}

\begin{proof}
We prove the claim for $A_0$; the argument for $-A_0$ is identical.
By Proposition~\ref{prop:multiplicity},
\begin{equation}
{\rm ord}_{A_0}(G_{A_0})=1.
\end{equation}
Choose $r>0$ sufficiently small such that $A_0$ is the only zero of
$G_{A_0}$ in $\overline{B(A_0,r)}$ and
\begin{equation}
\supp(P_{A_0}^\star)
\cap
[A_0-r,A_0+r]
=
{A_0}.
\end{equation}
By the holomorphic continuity in
Proposition~\ref{prop:continuity_vs_A} and
Lemma~\ref{lem:stability_lemma}, for all $A$ sufficiently close to
$A_0$, the disk $B(A_0,r)$ contains exactly one zero of $G_A$.
Since
\begin{equation}
A\in\supp(P_A^\star)
\end{equation}
and $A\in B(A_0,r)$ for $A$ sufficiently close to $A_0$, this zero is
$A$. Consequently,
\begin{equation}
\supp(P_A^\star)\cap(A_0-r,A_0+r)={A}.
\end{equation}

The boundary of $(A_0-r,A_0+r)$ has zero
$P_{A_0}^\star$-mass. Hence, by
Proposition~\ref{prop:continuity_vs_A} and the Portmanteau theorem,
\begin{align}
P_A^\star({A})
&=
P_A^\star\bigl((A_0-r,A_0+r)\bigr)\\
&\to
P_{A_0}^\star\bigl((A_0-r,A_0+r)\bigr)\\
&=
P_{A_0}^\star({A_0}).
\end{align}
The proof for the point $-A$ is analogous.
\end{proof}

\subsection{Support Size Transitions}
\label{sec:transtion_of_sup_points}

The following lemma converts the order of an isolated KKT zero into
a bound on the number of nearby support points.

\begin{lem}[Local support bound from zero multiplicity]
\label{lem:local_support_from_order}
Fix $A_0>0$, and suppose that zero is an isolated zero of $G_{A_0}$
of order
\begin{equation}
{\rm ord}_0 (G_{A_0})=m.
\end{equation}
Then there exist $0<r<A_0$ and $\varepsilon>0$ such that
\begin{equation}
\left|
\supp(P_A^\star)\cap(-r,r)
\right|
\leq
\frac{m}{2}
\label{eq:local_support_bound_from_order}
\end{equation}
for every $A>0$ satisfying
\begin{equation}
|A-A_0|<\varepsilon.
\end{equation}
\end{lem}

\begin{proof}
Choose $r\in(0,A_0/2)$ sufficiently small such that zero is the only
zero of $G_{A_0}$ in the closed disk
\begin{equation}
    \overline{B(0,r)}
    =
    \{z\in\bbC:|z|\leq r\}.
\end{equation}
By claim~4 of Proposition~\ref{prop:continuity_vs_A},
\begin{equation}
    G_A \to G_{A_0}
\end{equation}
uniformly on $\overline{B(0,r)}$ as $A\to A_0$. Therefore, by
Lemma~\ref{lem:stability_lemma}, $G_A$ has exactly $m$ zeros in
$B(0,r)$, counted according to multiplicity, for all $A$ sufficiently
close to $A_0$.

By taking $A$ sufficiently close to $A_0$, we also have $r<A$. Hence,
every real zero of $G_A$ in $(-r,r)$ is an interior zero and has even
multiplicity by Proposition~\ref{prop:multiplicity}. Let $ x_1,\ldots,x_\ell$ be the distinct real zeros of $G_A$ in $(-r,r)$. Then,
\begin{equation}
    2\ell
    \leq
    \sum_{j=1}^{\ell}{\rm ord}_{x_j}(G_A)
    \leq
    m.
\end{equation}
Consequently,
\begin{equation}
    \ell\leq\frac{m}{2}.
\end{equation}
Finally, since every support point of $P_A^\star$ is a real zero of
$G_A$, we obtain
\begin{equation}
    \left|
        \supp(P_A^\star)\cap(-r,r)
    \right|
    \leq
    \frac{m}{2},
\end{equation}
which proves \eqref{eq:local_support_bound_from_order} and concludes the proof. 
\end{proof}

 We now establish a local cardinality bound and characterize where changes in local support cardinality can occur. This result provides an unoriented local form of the transition conjecture of Sharma and Shamai~\cite{sharma2010transition}.

\begin{thm}
\label{thm:local_support_transition}
Fix $A_0>0$, and let  $K(A)=\left|\supp(P_A^\star)\right|$. Then there exists $\varepsilon>0$ such that
\begin{equation}
K(A_0) \le K(A)
\leq
K(A_0)+1
\label{eq:local_support_cardinality_bound}
\end{equation}
for every $A>0$ satisfying
\begin{equation}
|A-A_0|<\varepsilon.
\end{equation}

 Moreover, every nonzero support point
$x_0\in\supp(P_{A_0}^\star)$ has a unique local continuation
$x_A\in\supp(P_A^\star)$ such that
\begin{equation}
x_A\to x_0
\end{equation}
and
\begin{equation}
P_A^\star({x_A})
\to
P_{A_0}^\star({x_0})
\end{equation}
as $A\to A_0$. For the boundary points $x_0=\pm A_0$, the
corresponding continuations are $x_A=\pm A$. Hence all local changes in support cardinality and all branching or merging events are confined to a neighborhood of the origin.

\end{thm}

\begin{proof}
Let
\begin{equation}
    \supp(P_{A_0}^\star)
    =
    \{x_1,\ldots,x_{K(A_0)}\}.
\end{equation}
By Theorem~\ref{thm:size_of_degenerate_solution},
\begin{equation}
    \cD_{A_0}\cup\cI_{A_0}
    \subseteq
    \{0\}.
    \label{eq:only_origin_exceptional}
\end{equation}
Therefore, every nonzero interior support point of
$P_{A_0}^\star$ is nondegenerate. By
Proposition~\ref{prop:contininous_motion}, each such point has a
neighborhood containing exactly one support point of $P_A^\star$ for
all $A$ sufficiently close to $A_0$.

The boundary points $\pm A_0$ are handled by
Proposition~\ref{prop:boundary_continuity}: each gives rise to the
unique nearby boundary support point $\pm A$, respectively, and the
corresponding probability masses converge to those at amplitude
$A_0$.

It remains to consider a neighborhood of the origin. There are four possible cases.

\begin{itemize}

\item Suppose first that
\begin{equation}
    G_{A_0}(0)<0.
\end{equation}
Then $0 \notin \supp(P_{A_0}^\star)$.
By continuity, there exist $r>0$ and $\eta>0$ such that
\begin{equation}
    G_{A_0}(x)\leq-2\eta,
    \qquad x\in[-r,r].
\end{equation}
By claim~4 of Proposition~\ref{prop:continuity_vs_A}, for all $A$
sufficiently close to $A_0$,
\begin{equation}
    G_A(x)\leq-\eta,
    \qquad x\in[-r,r].
\end{equation}
Therefore, 
\begin{equation}
    \supp(P_{A}^\star) \cap (-r,r) = \varnothing.
\end{equation}
Thus, in this case the number of support points near the origin is unchanged and equal to zero.

\item Suppose that zero is a nondegenerate support point of
$P_{A_0}^\star$. Then, Proposition~\ref{prop:contininous_motion} provides $r>0$ such that, for all $A$ sufficiently close to $A_0$, 
\begin{equation}
    \left| \supp(P_A^\star) \cap (-r,r) \right| = 1.
\end{equation}
Thus, the single support point at the origin gives rise to exactly one nearby support point.

\item Suppose that
\begin{equation}
    0\notin\supp(P_{A_0}^\star)
    \qquad\text{and}\qquad
    G_{A_0}(0)=0.
\end{equation}
By Corollary~\ref{cor:order_at_origin},
\begin{equation}
    {\rm ord}_0(G_{A_0})=2.
\end{equation}
Therefore, Lemma~\ref{lem:local_support_from_order} gives $r>0$ such that, for all $A$ sufficiently close to $A_0$, 
\begin{equation}
    \left| \supp(P_A^\star) \cap (-r,r) \right| \le 1.
\end{equation}
Since the origin is not a support point at amplitude $A_0$, this neighborhood contributes zero points to $K(A_0)$ and at most one point to $K(A)$.

\item Finally, suppose that
\begin{equation}
    0\in\cD_{A_0}.
\end{equation}
By Corollary~\ref{cor:order_at_origin},
\begin{equation}
    {\rm ord}_0(G_{A_0})=4.
\end{equation}
 Hence, Lemma~\ref{lem:local_support_from_order} gives $r>0$ such that, for all $A$ sufficiently close to $A_0$, \begin{equation} 
\left| \supp(P_A^\star)\cap(-r,r) \right| \leq2. \label{eq:degenerate_origin_upper_local} 
\end{equation}

We now establish the corresponding lower bound. Choose $r>0$ sufficiently small that 
\begin{equation} 
\supp(P_{A_0}^\star)\cap[-r,r]=\{0\}. 
\end{equation} 
In particular, 
\begin{equation} 
P_{A_0}^\star((-r,r)) = P_{A_0}^\star(\{0\}) >0. \end{equation} 
Since the boundary points $\pm r$ have zero $P_{A_0}^\star$-mass, claim~1 of Proposition~\ref{prop:continuity_vs_A} and the Portmanteau theorem yield 
\begin{equation} 
P_A^\star((-r,r)) \to P_{A_0}^\star(\{0\}) >0 \qquad\text{as }A\to A_0. 
\end{equation} 
Consequently, 
\begin{equation} 
P_A^\star((-r,r))>0 
\end{equation} 
for all $A$ sufficiently close to $A_0$, and hence 
\begin{equation} 
\left| \supp(P_A^\star)\cap(-r,r) \right| \geq1. \label{eq:degenerate_origin_lower_local} 
\end{equation} 
Combining \eqref{eq:degenerate_origin_upper_local} and \eqref{eq:degenerate_origin_lower_local}, we obtain 
\begin{equation} 
1 \leq \left| \supp(P_A^\star)\cap(-r,r) \right| \leq2. 
\end{equation} 
Thus, the single support point at the origin can be replaced locally by either one or two support points, but it cannot disappear.

\end{itemize}

We now combine these local conclusions. Choose the neighborhoods above around every support point of $P_{A_0}^\star$ and around the origin, and shrink them, if necessary, so that their closures are pairwise disjoint. On the remaining compact subset of $[-A_0,A_0]$, the function $G_{A_0}$ has no zeros and, by the KKT inequality, is strictly negative. By claim~4 of Proposition~\ref{prop:continuity_vs_A}, $G_A$ remains strictly negative on this compact set for all $A$ sufficiently close to $A_0$. The boundary neighborhoods can also be chosen so as to contain the intervals between $\pm A_0$ and $\pm A$. Consequently, every support point of $P_A^\star$ lies in one of the selected neighborhoods. 

Let 
\begin{equation} 
s_0 = \mathbf{1}\{0\in\supp(P_{A_0}^\star)\}, 
\end{equation} 
where $\mathbf{1}\{\cdot\}$ is the indicator function.
There are exactly 
\begin{equation} 
K(A_0)-s_0 
\end{equation} 
nonzero support points at amplitude $A_0$, and each of them gives rise to exactly one nearby support point of $P_A^\star$. If 
\begin{equation} 
n_A = \left| \supp(P_A^\star)\cap(-r,r) \right| 
\end{equation} 
denotes the contribution of the neighborhood of the origin, the four cases above show that 
\begin{equation} 
s_0 \leq n_A \leq s_0+1. \label{eq:origin_local_count} 
\end{equation} 
Therefore, 
\begin{align} 
K(A_0) \le K(A_0)-s_0+n_A \leq  K(A_0)+1,
\end{align} 
and since $K(A) = K(A_0)-s_0+n_A$, we can conclude that
\begin{equation} 
K(A_0) \leq K(A) \leq K(A_0)+1 
\end{equation} 
for every $A$ sufficiently close to $A_0$, which proves \eqref{eq:local_support_cardinality_bound}. The preceding construction also shows that every nonzero support point has exactly one nearby continuation and that every possible change in the support configuration is confined to a neighborhood of the origin. This concludes the proof.

\end{proof}

\begin{rem} 
Theorem~\ref{thm:local_support_transition} shows that 
\begin{equation} 
K(A_0) \le K(A)\leq K(A_0)+1
\end{equation} 
for every $A$ sufficiently close to $A_0$. This local bound does not, however, establish that $K(A)$ is nondecreasing; in particular, it does not rule out a decrease in the support cardinality as $A$ increases. Thus, proving the full conjecture of Sharma and Shamai would additionally require showing that the map $A\mapsto K(A)$ is nondecreasing. This question remains open.
\end{rem}

 The next corollary characterizes the possible local changes in support
cardinality near the origin.

\begin{cor}
\label{cor:local_transition_geometry}
Fix $A_0>0$. Then there exist $r>0$ and $\varepsilon>0$ such that,
for every $A>0$ satisfying
\begin{equation}
    |A-A_0|<\varepsilon,
\end{equation}
all changes in the support configuration relative to
$P_{A_0}^\star$ are confined to $(-r,r)$. Moreover, the following
properties hold.

\begin{enumerate}
    \item If
    \begin{equation}
        \left|
        \supp(P_A^\star)\cap(-r,r)
        \right|=1,
    \end{equation}
    then
    \begin{equation}
        \supp(P_A^\star)\cap(-r,r)=\{0\}.
    \end{equation}

    \item If
    \begin{equation}
        \left|
        \supp(P_A^\star)\cap(-r,r)
        \right|=2,
    \end{equation}
    then there exists $x_A\in(0,r)$ such that
    \begin{equation}
        \supp(P_A^\star)\cap(-r,r)
        =
        \{-x_A,x_A\}.
    \end{equation}
\end{enumerate}

Consequently, locally in $A$, the only possible changes in the support
configuration at the origin are the appearance or disappearance of
the support point $0$, and the splitting or merging of the origin into
a symmetric pair.
\end{cor}

\begin{proof}
By Theorem~\ref{thm:local_support_transition}, after choosing $r>0$
and $\varepsilon>0$ sufficiently small, every nonzero support point of
$P_{A_0}^\star$ has exactly one nearby continuation, and no additional
support point can occur outside $(-r,r)$. Thus, every local change in
the support configuration is confined to $(-r,r)$.

Since the CAID is unique and
symmetric, its support is symmetric about the origin. Therefore, if
$\supp(P_A^\star)\cap(-r,r)$ consists of a single point, that point
must be invariant under reflection and hence must equal zero.

Similarly, if
$\supp(P_A^\star)\cap(-r,r)$ consists of exactly two points, symmetry
implies that they must be of the form $\{-x_A,x_A\}$ for some
$x_A\in(0,r)$.
\end{proof}

\begin{rem}
The transition geometry described in
Corollary~\ref{cor:local_transition_geometry} is not expected to be
universal. For different channel models or input constraints, more
general support transitions may occur. For example, the numerical
results for the noncoherent Rayleigh fading channel subject to
simultaneous average- and peak-power constraints in
\cite[Fig.~3]{favano2025rayleigh} suggest that two distinct positive
support branches can approach one another and merge as the
average-power constraint is varied. Hence, the local persistence of
nonzero support points established here is a specific structural
feature of the amplitude-constrained AWGN problem, rather than a
generic property of discrete capacity-achieving distributions.
\end{rem}

\section{Conclusion and Outlook}
In this work, we demonstrate that the multiplicity budget of the KKT function provides an effective tool for analyzing transitions in the support of the capacity-achieving input distribution. By combining a bound on the total number of zeros, counted with multiplicity, with the local geometry imposed by the KKT inequality, we show that the origin is the only possible inactive contact point or degenerate interior support point. Together with the continuity of the optimal input distribution and the KKT function, this yields the local stability of nonzero support points and shows that, locally in $A$, the support cardinality can increase by at most one, with any additional support point originating at the origin. Although the zero-counting argument used in this work exploits the Gaussian kernel, the underlying multiplicity analysis depends primarily on the KKT inequality and the analyticity of the corresponding KKT function. This suggests that the same approach may extend to other channels for which an appropriate zero-counting result and suitable regularity properties are available. Several interesting questions remain open. In particular, to establish the full conjecture of Sharma and Shamai \cite{sharma2010transition}, it remains to prove that the map $ A\mapsto K(A) $ is nondecreasing. The local transition theorem established here does not rule out a decrease in the support cardinality as $A$ increases.

Numerical evidence suggests that a considerably stronger property may
hold. Let
\begin{equation}
    \boldsymbol{p}(A)
    =
    \bigl(
        p_{A,1}^{\downarrow},
        p_{A,2}^{\downarrow},
        \ldots
    \bigr)
\end{equation}
denote the probability masses of $P_A^\star$, arranged in nonincreasing
order and padded with zeros. Numerical simulations indicate that, for
$0<A_1<A_2$, these vectors appear to satisfy
\begin{equation}
    \boldsymbol{p}(A_2)
    \prec
    \boldsymbol{p}(A_1),
    \label{eq:numerical_majorization_property}
\end{equation}
where $\prec$ denotes majorization. Equivalently, let
\begin{equation}\label{eq:partial_sums}
    S_r(A)=\sum_{j=1}^r p_{A,j}^{\downarrow}.
\end{equation}
Then, \eqref{eq:numerical_majorization_property} is equivalent to
\begin{equation}
    S_r(A_2)\leq S_r(A_1),
    \qquad r\geq1.
\end{equation}
Thus, the probability masses appear to become progressively less
concentrated as the amplitude constraint increases. Fig.~\ref{fig:majorization} provides numerical evidence for the validity of \eqref{eq:numerical_majorization_property}. 

\begin{figure}
\centering
\input{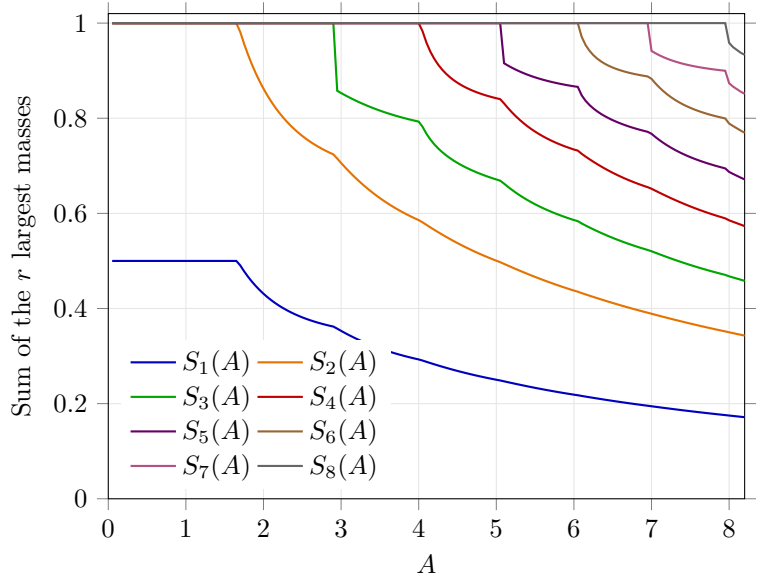}
\caption{Plot of partial sums in \eqref{eq:partial_sums} vs.~$A$.}
\label{fig:majorization}
\end{figure}

If this
majorization property could be established, the majorization property of R\'enyi entropies  of nonnegative order \cite{sason2018tight}, would imply that all
R\'enyi entropies  of the optimal input distribution are nondecreasing
in $A$. In particular, since
\begin{equation}
    H_0(P_A^\star)=\log K(A),
\end{equation}
it would immediately yield the desired monotonicity of $K(A)$.

\appendices

\section{Proof of Lemma~\ref{lem:roll_mulitplicity}}
\label{app:lem:roll_mulitplicity}

Since $x_i$ is a zero of $f$ of multiplicity $m_i$, we have
\begin{equation}
f(x_i)=f'(x_i)=\cdots=f^{(m_i-1)}(x_i)=0,
\qquad
f^{(m_i)}(x_i)\neq 0.
\end{equation}
Therefore, $x_i$ is a zero of $f'$ of multiplicity exactly
$m_i-1$. These zeros contribute
\begin{equation}
\sum_{i=1}^r (m_i-1)
\end{equation}
to the number of zeros of $f'$, counted with multiplicity.

In addition, for every $i=1,\ldots,r-1$, we have
\begin{equation}
f(x_i)=f(x_{i+1})=0.
\end{equation}
By Rolle's theorem, there exists a point
\begin{equation}
c_i\in(x_i,x_{i+1})
\end{equation}
such that
\begin{equation}
f'(c_i)=0.
\end{equation}
The intervals $(x_i,x_{i+1})$ are disjoint, so these provide at least
$r-1$ additional distinct zeros of $f'$.

Consequently,
\begin{align}
N_{\mathrm{mul}}(f')
&\geq
\sum_{i=1}^r(m_i-1)+(r-1)\\
&=
\left(\sum_{i=1}^r m_i\right)-1\\
&=
N_{\mathrm{mul}}(f)-1.
\end{align}

\section{Proof of Proposition~\ref{prop:continuity_vs_A}}
\label{app:proof_continuity}

\subsection{Proof of Continuity of the Input and Output
Distributions and Capacity}

Since $A_n\to A_0$, the sequence $\{A_n\}$ is bounded. Hence, there
exists $B<\infty$ such that
\begin{equation}
    \supp(P_{A_n}^\star)\subseteq[-B,B],
    \qquad n\geq1.
    \label{eq:common_compact_support}
\end{equation}
Therefore, the sequence $\{P_{A_n}^\star\}$ is tight.\footnote{A sequence of probability measures $\{P_n\}$ on $\bbR$ is
said to be tight if, for every $\varepsilon>0$, there exists a compact
set $\mathcal K\subset\bbR$ such that
$P_n(\mathcal K)\geq1-\varepsilon$ for every $n$. Equivalently, there
exists $M<\infty$ such that
$\sup_n P_n(\{x:|x|>M\})<\varepsilon$.} By Prokhorov's
theorem \cite[Thm.~9.6.2]{resnick1999probability}, every
subsequence has a further weakly convergent subsequence.

Consider an arbitrary subsequence of $\{P_{A_n}^\star\}$ and extract
a further subsequence, indexed by $\{n_k\}$, such that
\begin{equation}
    P_{A_{n_k}}^\star
    \Rightarrow
    P
    \label{eq:subsequential_input_limit}
\end{equation}
for some probability distribution $P$. We first show that $P$ is
feasible at amplitude $A_0$. Fix $\delta>0$. Since
$A_{n_k}\to A_0$, for all sufficiently large $k$,
\begin{equation}
    \supp(P_{A_{n_k}}^\star)
    \subseteq
    [-A_0-\delta,A_0+\delta].
\end{equation}
Since $[-A_0-\delta,A_0+\delta]$ is closed, the Portmanteau theorem
and \eqref{eq:subsequential_input_limit} imply
\begin{equation}
    P([-A_0-\delta,A_0+\delta])=1.
\end{equation}
Since $\delta>0$ is arbitrary, it follows that
\begin{equation}
    P([-A_0,A_0])=1.
\end{equation}
Thus, $P$ is feasible for the capacity problem at amplitude $A_0$.

We next show that $P$ is capacity achieving at amplitude $A_0$.
Define
\begin{equation}
    s_n
    =
    \min\left\{1,\frac{A_n}{A_0}\right\},
\end{equation}
and let $\widetilde P_n$ denote the distribution of
$s_nX_{A_0}^\star$. By construction, $\widetilde P_n$ is feasible at
amplitude $A_n$. Moreover, since $s_n\to1$,
\begin{equation}
    \widetilde P_n
    \Rightarrow
    P_{A_0}^\star.
    \label{eq:scaled_optimal_input_convergence}
\end{equation}

For the scalar Gaussian channel, the mutual-information functional is
continuous, with respect to the L\'evy metric, on the set of input
distributions supported on a fixed compact interval
\cite{smith1971information}. Since the L\'evy metric metrizes weak
convergence on $\bbR$, and all the distributions considered here are
supported on a common compact interval, weak convergence implies
convergence of the corresponding mutual informations. Therefore,
from \eqref{eq:scaled_optimal_input_convergence},
\begin{equation}
    I(\widetilde P_n)
    \to
    I(P_{A_0}^\star)
    =
    C(A_0).
    \label{eq:scaled_MI_convergence}
\end{equation}
Since $\widetilde P_n$ is feasible at amplitude $A_n$,
\begin{equation}
    C(A_n)
    \geq
    I(\widetilde P_n),
\end{equation}
and consequently
\begin{equation}
    \liminf_{n\to\infty}C(A_n)
    \geq
    C(A_0).
    \label{eq:capacity_lower_continuity}
\end{equation}

On the other hand, applying the same continuity result to
\eqref{eq:subsequential_input_limit} gives
\begin{align}
    I(P)
    &=
    \lim_{k\to\infty}
    I(P_{A_{n_k}}^\star)\\
    &=
    \lim_{k\to\infty}
    C(A_{n_k}).
    \label{eq:subsequence_capacity_limit}
\end{align}
Hence,
\begin{align}
    I(P)
    &=
    \lim_{k\to\infty}C(A_{n_k})\\
    &\geq
    \liminf_{n\to\infty}C(A_n)\\
    &\geq
    C(A_0),
    \label{eq:limit_distribution_capacity_lower}
\end{align}
where the last inequality follows from
\eqref{eq:capacity_lower_continuity}. Since $P$ is feasible at
amplitude $A_0$, the definition of capacity also gives
\begin{equation}
    I(P)\leq C(A_0).
\end{equation}
Therefore,
\begin{equation}
    I(P)=C(A_0).
\end{equation}
Thus, $P$ is capacity achieving at amplitude $A_0$. By uniqueness of
the CAID
\cite{smith1971information},
\begin{equation}
    P=P_{A_0}^\star.
    \label{eq:subsequential_limit_unique}
\end{equation}

Since the original subsequence was arbitrary, every subsequence of
$\{P_{A_n}^\star\}$ has a further subsequence converging weakly to
$P_{A_0}^\star$. It follows that the entire sequence converges:
\begin{equation}
    P_{A_n}^\star
    \Rightarrow
    P_{A_0}^\star.
    \label{eq:input_weak_continuity_final}
\end{equation}
Indeed, otherwise there would exist a weak neighborhood $\mathcal U$
of $P_{A_0}^\star$ and a subsequence lying entirely outside
$\mathcal U$; that subsequence could not have a further subsequence
converging to $P_{A_0}^\star$, contradicting
\eqref{eq:subsequential_limit_unique}.

We also show the resulting continuity of the capacity, which will
be used below. By \eqref{eq:input_weak_continuity_final} and the
continuity of the mutual-information functional on the common compact
support,
\begin{align}
    C(A_n)
    &=
    I(P_{A_n}^\star)\\
    &\to
    I(P_{A_0}^\star)\\
    &=
    C(A_0).
    \label{eq:capacity_continuity}
\end{align}

It remains to establish weak convergence of the induced output
distributions. Let $g:\bbR\to\bbR$ be bounded and continuous, and
define
\begin{equation}
    T_g(x)
    :=
    \bbE[g(x+Z)]
    =
    \int_{\bbR}g(x+z)\phi(z)\rmd z.
\end{equation}
The function $T_g$ is bounded. It is also continuous: if $x_m\to x$,
then
\begin{equation}
    g(x_m+z)\to g(x+z)
\end{equation}
for every $z$, and dominated convergence yields
\begin{equation}
    T_g(x_m)\to T_g(x).
\end{equation}
Therefore, by \eqref{eq:input_weak_continuity_final},
\begin{align}
    \int_{\bbR} g(y)\,P_{Y,A_n}^\star(\rmd y)
    &=
    \int_{\bbR} T_g(x)\,P_{A_n}^\star(\rmd x)\\
    &\to
    \int_{\bbR} T_g(x)\,P_{A_0}^\star(\rmd x)\\
    &=
    \int_{\bbR} g(y)\,P_{Y,A_0}^\star(\rmd y).
\end{align}
Hence,
\begin{equation}
    P_{Y,A_n}^\star
    \Rightarrow
    P_{Y,A_0}^\star.
\end{equation}
This proves the first claim of Proposition~\ref{prop:continuity_vs_A}.

\subsection{Proof of Convergence of the Output Densities}

For every $y\in\bbR$,
\begin{equation}
    f_{Y,A_n}^\star(y)
    =
    \int_{\bbR}\phi(y-x)\,P_{A_n}^\star(\rmd x).
\end{equation}
Since $x\mapsto\phi(y-x)$ is bounded and continuous, the weak
convergence $P_{A_n}^\star\Rightarrow P_{A_0}^\star$ gives
\begin{equation}
    f_{Y,A_n}^\star(y)
    \to
    f_{Y,A_0}^\star(y).
\end{equation}
Moreover,
\begin{equation}
    \left|
        \frac{\rmd}{\rmd y}f_{Y,A_n}^\star(y)
    \right|
    \leq
    \sup_{t\in\bbR}|\phi'(t)|
    <\infty.
\end{equation}
Thus, the sequence $\{f_{Y,A_n}^\star\}$ is uniformly Lipschitz and,
hence, equicontinuous on $\bbR$. Since it converges pointwise to
$f_{Y,A_0}^\star$, the Arzel\'a-Ascoli theorem implies that the
convergence is uniform on every compact subset of $\bbR$; see, e.g.,
\cite[Thm.~4.44]{folland2013real}. Consequently,
\begin{equation}
    f_{Y,A_n}^\star
    \to
    f_{Y,A_0}^\star
\end{equation}
locally uniformly.

\subsection{Convergence of $G_A''$} 

Since $A_n\to A_0$, the distributions $P_{A_n}^\star$ are supported
on a common compact interval $[-B,B]$.  For $k=1,2$, define
\begin{equation}
    N_{k,A}(y)
    =
    \int_{\bbR}t^k\phi(y-t)\,P_A^\star(\rmd t),
\end{equation}
so that
\begin{equation}
    M_{k,A}(y)
    =
    \frac{N_{k,A}(y)}{f_{Y,A}^\star(y)}.
\end{equation}
Let $\mathcal L\subset\bbR$ be compact. Since the distributions
$P_{A_n}^\star$ are supported on the common compact interval $[-B,B]$,
the function
\begin{equation}
    (t,y)\mapsto t^k\phi(y-t)
\end{equation}
is bounded and uniformly continuous on $[-B,B]\times\mathcal L$.
Consequently, weak convergence $P_{A_n}^\star\Rightarrow
P_{A_0}^\star$, together with a finite-net argument on $\mathcal L$,
implies that
\begin{equation}
    \sup_{y\in\mathcal L}
    \left|N_{k,A_n}(y)-N_{k,A_0}(y)\right|
    \to 0.
\end{equation}
Moreover, $f_{Y,A_0}^\star$ is continuous and strictly positive.
Therefore,
\begin{equation}
    \inf_{y\in\mathcal L}f_{Y,A_0}^\star(y)>0.
\end{equation}
Since $f_{Y,A_n}^\star\to f_{Y,A_0}^\star$ uniformly on
$\mathcal L$, it follows that
\begin{equation}
    M_{k,A_n}\to M_{k,A_0},
    \qquad k=1,2,
\end{equation}
uniformly on $\mathcal L$. Therefore, if
\begin{equation}
    {\rm Var}(X_{A_n}^\star\mid Y_{A_n}^\star=y)
    =
    M_{2,{A_n}}(y)-M_{1,{A_n}}(y)^2,
\end{equation}
then
\begin{equation}
    {\rm Var}(X_{A_n}^\star\mid Y_{A_n}^\star=y)\to {\rm Var}(X_{A_0}^\star\mid Y_{A_0}^\star=y)
\end{equation}
locally uniformly. Moreover, since $|X_{A_n}^\star|\leq B$,
\begin{equation}
    0\leq {\rm Var}(X_{A_n}^\star\mid Y_{A_n}^\star=y) \leq B^2.
\end{equation}

From Lemma~\ref{lem:Derivative},
\begin{equation}
    G_A''(x)
    =
    1-\int_{\bbR}\phi(y-x) {\rm Var}(X_A^\star\mid Y_A^\star=y)\rmd y.
\end{equation}
Hence, for every compact set $\mathcal K\subset\bbR$,
\begin{align}
    \sup_{x\in\mathcal K}
    |G_{A_n}''(x)-G_{A_0}''(x)|
    &\leq
    \int_{\bbR}
        \sup_{x\in\mathcal K}\phi(y-x)
        \left|{\rm Var}(X_{A_n}^\star\mid Y_{A_n}^\star=y)-{\rm Var}(X_{A_0}^\star\mid Y_{A_0}^\star=y) \right|
    \rmd y \to 0,
\end{align}
where the last step follows from dominated convergence. Consequently,
\begin{equation}
    G_{A_n}''\to G_{A_0}''
\end{equation}
uniformly on every compact subset of $\bbR$.

\subsection{Proof of Holomorphic Continuity of $G_A$}
Let $\mathcal K\subset\bbC$ be compact.
By Lemma~\ref{lem:entire_extension_GA}, the KKT function admits the
entire extension
\begin{equation}
    G_A(z)
    =
    -h(Z)-C(A)
    -
    \int_{\bbR}
        \phi(y-z)\log f_{Y,A}^\star(y)
    \rmd y,
    \qquad z\in\bbC.
    \label{eq:complex_extension_GA_continuity}
\end{equation}

Since $A_n\to A_0$, there exists $B<\infty$ such that
\begin{equation}
    \supp(P_{A_n}^\star)\cup\supp(P_{A_0}^\star)
    \subseteq[-B,B]
\end{equation}
for every sufficiently large $n$. Therefore,
\begin{equation}
    \phi(|y|+B)
    \leq
    f_{Y,A_n}^\star(y)
    \leq
    \frac{1}{\sqrt{2\pi}},
    \qquad y\in\bbR,
\end{equation}
and the same bounds hold for $f_{Y,A_0}^\star$. Consequently, there
exists $c_B<\infty$ such that
\begin{equation}
    \left|\log f_{Y,A_n}^\star(y)\right|
    +
    \left|\log f_{Y,A_0}^\star(y)\right|
    \leq
    c_B(1+y^2).
    \label{eq:uniform_log_density_bound}
\end{equation}

Let
\begin{equation}
    M_{\mathcal K}
    =
    \sup_{z\in\mathcal K}|z|.
\end{equation}
Writing $z=u+\rmi v$, we have
\begin{align}
    |\phi(y-z)|
    &=
    \frac{1}{\sqrt{2\pi}}
    \exp\left(
        -\frac{(y-u)^2}{2}+\frac{v^2}{2}
    \right)\\
    &\le \frac{1}{\sqrt{2\pi}}
    \exp\left(
        -\frac{y^2}{4}+\frac{|z|^2}{2}
    \right)\\
    &\leq
    c_{\mathcal K}\exp\left(-\frac{y^2}{4}\right),
    \qquad z\in\mathcal K,
    \label{eq:complex_Gaussian_compact_bound}
\end{align}
where
\begin{equation}
    c_{\mathcal K}
    =
    \frac{\exp(M_{\mathcal K}^2/2)}{\sqrt{2\pi}}.
\end{equation}

By the locally uniform convergence of the output densities,
\begin{equation}
    \log f_{Y,A_n}^\star(y)
    \to
    \log f_{Y,A_0}^\star(y),
    \qquad y\in\bbR.
    \label{eq:log_output_density_convergence}
\end{equation}
Moreover, the first part of the proof gives
\begin{equation}
    C(A_n)\to C(A_0).
    \label{eq:capacity_continuity_used_complex}
\end{equation}
Using \eqref{eq:complex_extension_GA_continuity}, we obtain
\begin{align}
    \sup_{z\in\mathcal K}
    |G_{A_n}(z)-G_{A_0}(z)|
    &\leq
    |C(A_n)-C(A_0)| \notag\\
    &\quad+
    \int_{\bbR}
        \sup_{z\in\mathcal K}|\phi(y-z)|
        \left|
            \log f_{Y,A_n}^\star(y)
            -
            \log f_{Y,A_0}^\star(y)
        \right|
    \rmd y.
    \label{eq:complex_GA_difference_bound}
\end{align}
By \eqref{eq:uniform_log_density_bound} and
\eqref{eq:complex_Gaussian_compact_bound}, the integrand in
\eqref{eq:complex_GA_difference_bound} is bounded by
\begin{equation}
    c_{B,\mathcal K}(1+y^2)
    \exp\left(-\frac{y^2}{4}\right),
\end{equation}
which is integrable and independent of $n$. Thus, by
\eqref{eq:log_output_density_convergence} and the dominated convergence
theorem,
\begin{equation}
    \sup_{z\in\mathcal K}
    |G_{A_n}(z)-G_{A_0}(z)|
    \to 0.
\end{equation}
Since $\mathcal K\subset\bbC$ was arbitrary, the convergence is uniform
on every compact subset of $\bbC$.

\bibliography{refs.bib}
\bibliographystyle{IEEEtran}

\end{document}